\documentclass[10.5pt]{article}
\usepackage{amsfonts}
\usepackage[leqno]{amsmath}
\usepackage{graphicx}
\usepackage{latexsym}
\usepackage{amsmath,amsfonts,amssymb,amsthm,mathrsfs,euscript,makeidx,color}
\usepackage{enumerate}

\usepackage[english]{babel}
\allowdisplaybreaks[1]

\usepackage[utf8]{inputenc}
\usepackage{selinput}

\usepackage{babel} %
\SelectInputMappings{%
	agrave={à},
	ccedilla={ç},
	Euro={€}
}

\usepackage{multirow}

\def\sqr#1#2{{\vcenter{\vbox{\hrule height.#2pt
              \hbox{\vrule width.#2pt height#1pt \kern#1pt \vrule width.#2pt}
              \hrule height.#2pt}}}}
\def\5n{\negthinspace \negthinspace \negthinspace \negthinspace \negthinspace }
\def\4n{\negthinspace \negthinspace \negthinspace \negthinspace }
\def\3n{\negthinspace \negthinspace \negthinspace }
\def\2n{\negthinspace \negthinspace }
\def\1n{\negthinspace }

\def\dbE{\mathbb{E}}
\def\dbF{\mathbb{F}}

\def\dbP{\mathbb{P}}

\def\dbR{\mathbb{R}}

\def\cD{{\cal D}}

\def\cF{{\cal F}}

\def\cU{{\cal U}}

\def\ds{\displaystyle}

\def\ns{\noalign{\ss}}

\def\ss{\smallskip}
\def\ms{\medskip}
\def\bs{\bigskip}
\def\q{\quad}
\def\qq{\qquad}
\def\hb{\hbox}

\def\({\Big (}
\def\){\Big )}
\def\[{\Big[}
\def\]{\Big]}

\def\rf{\eqref}

\def\a{\alpha}

\def\g{\gamma}

\def\e{\varepsilon}

\def\l{\lambda}

\def\si{\sigma}
\def\t{\tau}
\def\f{\varphi}

\def\i{\infty}
\def\G{\Gamma}
\def\D{\Delta}
\def\Th{\Theta}

\def\O{\Omega}

\def\limsup{\mathop{\overline{\rm lim}}}

\def\h{\widehat}

\def\ti{\tilde}
\def\cd{\cdot}
\def\cds{\cdots}

\def\deq{\triangleq}
\def\les{\leqslant}
\def\ges{\geqslant}

\def\rf{\eqref}
\def\bde{\begin{definition}\label}
\def\ede{\end{definition}}
\def\bel{\begin{equation}\label}
\def\bt{\begin{theorem}\label}
\def\et{\end{theorem}}
\def\bc{\begin{corollary}\label}
\def\ec{\end{corollary}}
\def\bl{\begin{lemma}\label}
\def\el{\end{lemma}}
\def\bp{\begin{proposition}\label}
\def\ep{\end{proposition}}
\def\br{\begin{remark}\label}
\def\er{\end{remark}}
\def\bex{\begin{example}\label}
\def\ex{\end{example}}
\def\ben{\begin{enumerate}}
\def\een{\end{enumerate}}

\def\square#1{\vbox{\hrule\hbox{\vrule height#1%
     \kern#1\vrule}\hrule}}
\def\rectangle#1#2{\vbox{\hrule\hbox{\vrule height#1%
     \kern#2\vrule}\hrule}}
\def\qed{\hfill \vrule height7pt width3pt depth0pt}

\font\tenbb=msbm10 \font\sevenbb=msbm7 \font\fivebb=msbm5

\newfam\bbfam
\scriptscriptfont\bbfam=\fivebb \textfont\bbfam=\tenbb
\scriptfont\bbfam=\sevenbb

\makeatletter
   
   \@addtoreset{equation}{section}
\makeatother

\usepackage{amsmath, amsfonts,amssymb, amsthm, euscript,makeidx,color,mathrsfs,latexsym,bm}
\usepackage{graphicx}
\usepackage{enumitem}

\usepackage{natbib}
           \makeatletter
           \renewcommand{\bibsection}{
           \begin{center}

\section*{\refname\@mkboth{\MakeUppercase{\refname}}
               {\MakeUppercase{\refname}}}
           \end{center}
           }
           \makeatother

\newtheorem{proposition}{Proposition}
\newtheorem{definition}{Definition}
\newtheorem{theorem}{Theorem}
\newtheorem{corollary}{Corollary}
\newtheorem{lemma}{Lemma}
\newtheorem{remark}{Remark}
\newtheorem{example}{Example}

\newtheorem{assumption}{Assumption}

\newcommand{\ba}{\begin{array}}
\newcommand{\ea}{\end{array}}
\newcommand{\be}{\begin{equation}}
\newcommand{\ee}{\end{equation}}
\newcommand{\bee}{\begin{equation*}}
\newcommand{\eee}{\end{equation*}}
\newcommand{\bea}{\begin{eqnarray}}
\newcommand{\eea}{\end{eqnarray}}
\newcommand{\beaa}{\begin{eqnarray*}}
\newcommand{\eeaa}{\end{eqnarray*}}

\def\dbE{\mathbb{E}}
\def\dbF{\mathbb{F}}

\def\dbP{\mathbb{P}}
\def\dbR{\mathbb{R}}

\def\a{\alpha}

\def\g{\gamma}

\def\e{\varepsilon}

\def\l{\lambda}

\def\si{\sigma}
\def\t{\tau}
\def\f{\varphi}

\def\h{\widehat}
\def\G{\Gamma}
\def\D{\Delta}
\def\Th{\Theta}

\def\O{\Omega}
\def\cD{{\cal D}}

\def\cF{{\cal F}}

\def\cU{{\cal U}}

\def\ss{\smallskip}
\def\ms{\medskip}
\def\bs{\bigskip}
\def\q{\quad}
\def\qq{\qquad}
\def\hb{\hbox}

\def\cd{\cdot}
\def\cds{\cdots}

\newcommand{\basa}{\begin{assumption}}
\newcommand{\easa}{\end{assumption}}

\newcommand{\bas}{\begin{assum}}
\newcommand{\eas}{\end{assum}}

\def\limsup{\mathop{\overline{\rm lim}}}

\def\h{\widehat}

 \def\cd{\cdot}
\def\cds{\cdots}

\def\deq{\mathop{\buildrel\D\over=}}

\def\1{{\bf 1}}

\def\:{\!:\!}
\def\reff#1{{\rm(\ref{#1})}}

at 9pt

\newtheorem{thm}{Theorem}
\newtheorem{lem}{Lemma}
\newtheorem{prop}{Proposition}
\newtheorem{defn}{Definition}
\newtheorem{assum}{Assumption}

\stepcounter{equation}
\usepackage{chngcntr}
\counterwithout{equation}{section}

\begin{document}
\title{\bf Present-Biased Lobbyists in Linear Quadratic Stochastic Differential Games\footnote{We thank Dongliang Lu for outstanding research assistance.}}

\author{Ali Lazrak\footnote{Sauder School of
Business, University of British Columbia, 2053 Main Mall, Vancouver, BC V6T 1Z2, Canada; (Email: {tt ali.lazrak@sauder.ubc.ca}. This author is supported by SSHRC.},~~~~Hanxiao Wang\footnote{ College of Mathematics and Statistics, Shenzhen University, Shenzhen 518060, China
(Email: {\tt hxwang@szu.edu.cn}). This author is supported in part by NSFC Grant 12201424 and
Guangdong Basic and Applied Basic Research Foundation 2023A1515012104.},~~~~Jiongmin Yong\footnote{Department of
Mathematics, University of Central Florida, Orlando, FL 32816, USA (Email: {\tt jiongmin.yong@ucf.edu}).
This author was supported in part by NSF Grant DMS-2305475.}}

\date{}
\maketitle

\bs

\centerline{({\it In the Memory of Professor Tomas Bj\"ork})}

\bs

\begin{abstract}
We investigate a linear quadratic stochastic zero-sum game where two players lobby a political representative to invest in a wind turbine farm. Players are time-inconsistent because they discount performance with a non-constant rate. Our objective is to identify a consistent planning equilibrium in which the players are aware of their inconsistency and cannot commit to a lobbying policy. We analyze the equilibrium behavior in both single player and two-player cases, and compare the behavior of the game under constant and non-constant discount rates. The equilibrium behavior is provided in closed-loop form, either analytically or via numerical approximation. Our numerical analysis of the equilibrium reveals that strategic behavior leads to more intense lobbying without resulting in overshooting.

\end{abstract}

\bf Keywords. \rm Time inconsistency, lobbying, two players zero sum dynamic game, linear quadratic stochastic differential game

\ms

\bf AMS 2020 Mathematics Subject Classification. \rm 91A15,~91B14

\section{Introduction}

Time inconsistency refers to a phenomenon in which a decision maker's preferences for different alternatives change over time, even in the absence of new information. This can pose a significant challenge in solving dynamic optimal choice problems, as the optimal solution may vary depending on the moment in time from which the decision is being made. This misalignment can create a gap between the optimal policies intended at one point in time and the policies that are implemented at a later time. Standard neoclassical models in macroeconomics and finance assume that decision makers have time-additive preferences that obey the independence axiom and discount utility exponentially.  In that context, dynamic programming methods can help break down any dynamic optimization problem into a sequence of simpler, static problems that can be solved recursively. However, economists are acutely aware that time inconsistency pervades various contexts, rendering the principle of dynamic programming invalid. Time inconsistency can manifest itself in common economic interactions, even if the players' preferences are standard. For instance, in dynamic games, the time inconsistency problem is a common issue, and the players' ability to commit can significantly affect the equilibrium outcome. The game between central banks and the private sector in \cite{kydland1977rules} is an example of this problem. Similarly, in dynamic collective decisions, \cite{jackson2015collective} have shown that time inconsistency is prevalent for almost any collective decision rule, whether it is vote-based or based on some utilitarian aggregation rule.\footnote{\cite{calvo1988optimal} and \cite{bernheim1989intergenerational} have also highlighted the problem of planners' time inconsistency in intergenerational models.} Because it emerges from the interaction between players, the time inconsistency problem can provide new insights and allow economists to recommend welfare-improving policies for society (\cite{kydland1977rules}). Overall, the issue of time inconsistency is central among economists and has significant policy implications.

Time inconsistency can also manifest in single individual decisions in cases where decision makers deviate from the standard assumption of exponential discounting. When a non-exponential discount function is used, the marginal rate of substitution between consumption at two different future dates varies then as time passes, unlike in the case of a constant discount rate. Discount rates can vary according to a prominent psychological theory proposed by \cite{ainslie1975specious}. The theory suggests that hyperbolic discounting generates a bias for the present and can explain the impulsive behavior explored in \cite{ainslie1975specious}.  Hyperbolic discounting refers to a phenomenon where people tend to discount the value of future rewards more heavily when they are further away in time, but less so as the reward becomes more immediate. 


This paper investigates the impact of a non-constant discount rate on continuous-time linear quadratic optimization problem. We assume that the decision-maker is aware of the time inconsistency problem and approaches the problem with an intrapersonal game view while lacking any commitment ability. This is the consistent planning solution to the game envisioned by  \cite{strotz1956myopia}. To contrast the behavior of a decision-maker with a constant discount rate versus one with a non-constant discount rate, we consider two variations of the same problem. In the first variation, we examine the behavior of a single decision-maker. In the second variation, we analyze a zero-sum game between two players. To gain insights into the impact of time inconsistency on decision-making in these two variations, we aim to provide closed-form solutions or numerical approximations. This will enable us to offer concrete statements about decision-making behavior and contrast it with the behavior under constant rate discounting. It is crucial to underline that our examples and solution methods are exclusively applicable to zero-sum games.

The framework of this paper deliberately maintains simplicity, serving as an entry point to relevant literature and offering insights into how changing discount rates affect decision-making behavior.
Our goal is to present closed-form solutions wherever possible and, when not feasible, to provide numerical approximation methods. By prioritizing simplicity and accessibility, we aim to engage a wider audience in this crucial area of research. Notably, Section~\ref{section: zero sum with non constant discount} presents new findings on the zero-sum game with players exhibiting non-constant discount rates, which have not been explored to the best of our knowledge.

{\it Discussion of the literature.} The topic of time inconsistency has been extensively studied in the literature, and it is not feasible to provide a comprehensive summary in this paper. Therefore, we do not aim to cite every significant paper in this field.\footnote{The book by \cite{bjork2021time} and the recent paper by \cite{hernandez2023myself} offer comprehensive overviews of the literature on time inconsistency in mathematical finance.} The seminal paper by \cite{strotz1956myopia} was the first to formalize non-exponential discounting within a dynamic utility maximization framework. In a simple deterministic ``cake eating" problem in continuous time, Strotz formulated the solution of the problem under commitment, representing the optimal solution based on preferences at the beginning of time. He then went on by solving the consistent planning problem that he defined as:

 ``Since precommitment is not always a feasible solution to the problem of intertemporal conflict, the man with insight into his future unreliability may adopt a different strategy and reject any plan which he will not follow through. His problem is then to find the best plan among those that he will actually  follow." (\cite{strotz1956myopia} page 173).

The consistent planning solution, also known as the sophisticated solution, provides a useful framework for solving consumption and saving decisions over time when discount rates are non-constant. In a finite horizon setting with discrete time, the consistent planning solution can be identified by using backward induction. Starting from the last period, the decision maker maximizes their utility by selecting a sustainable consumption plan that future selves would also choose. This is done at each step, working backwards to the present time. In each step of the consistent planning solution, the lifetime utility is maximised under a consumption sustainability constraint in addition to the standard budget constraint, making the problem non-standard.  The sustainability constraints ensures that the needs of future selves are taken into account. Although this method may generate multiple solutions due to the potential non-concavity of intermediate values, the existence of a solution is typically guaranteed.

When the time horizon is infinite, there is no terminal time to start the backward induction and the consistent planning solution can present some mathematical difficulties in discrete time. Nonetheless, researchers have made significant progress in addressing these difficulties, as seen in \cite{krusell2003consumption}, among others. In the original continuous framework proposed by Strotz, the mathematical formulation of consistent planning was initially established in the context of a deterministic consumption-saving problem by \cite{ekeland2006noncommitment} (see also \cite{ekeland2010golden}). The approach involves assuming that the decision maker has control over the consumption of their immediate successors at any given point in time, which enables the formation of a small coalition that isolates the current decision maker from the more distant ones. For a closed-loop consumption strategy that depends on the current value of capital stock to qualify as an equilibrium strategy, it must be the optimal policy for the current planner when the coalition is infinitesimally small. Furthermore, the same strategy must be expected to be employed by the distant planners.

This way of defining the equilibrium is commonly referred to as the ``spike variation method". Interestingly, this method is more general than its original framework proposed by Ekeland and Lazrak (2006). The equilibrium is well-posed even when a Brownian noise is introduced in different applications. Early research on this topic includes Ekeland and Pirvu (2008), Bj\"ork and Murgoci (2010), Yong (2012), Hu et al. (2012), and Bj\"ork et al. (2017). This paper presents an application of the spike variation method to a two-player, zero-sum stochastic linear quadratic differential game. Our approach represents a step forward in the understanding of strategic interactions between players when each player has a preference based time inconsistency. It is worth noting that recent progress has been made in principal-agent models concerning the resolution of problems involving non-constant discounting, as demonstrated in \cite{cetemen2023renegotiation} and \cite{hernandez2023time}. Additionally, in the context of optimal control with multidimensional states, and the control appears in the diffusion, \cite{lei2021nonlocality} have established the time-local well-posedness of the equilibrium HJB equation when dealing with time-inconsistent single player situation. Given the focus of our study, we concentrate on specific examples to illustrate the impact of time inconsistency in zero sum differential games. This approach allows us to define an equilibrium that is global in the time dimension, enabling us to derive valuable economic insights from our findings.




 \section{The model}

Let $(\O,\cF,\dbF,\dbP)$ be a complete filtered probability space on which a one-dimensional standard Brownian motion $W(\cd)$ is defined, whose natural filtration, augmented by all the $\dbP$-null set in $\cF$, is denoted by $\dbF=\{\cF_t\}_{t\ges0}$. Let $T>0$ be a fixed time horizon and define the set
$$\D^*[0,T]=\{(t,s)\in[0,T]^2\bigm|0\les t\les s\les T\}.$$
{\bf Objectives and policies.} We consider a dynamic game where a state is controlled by two players.  The state is described by a one-dimensional controlled linear stochastic differential equation given by
\bel{SDE2}\left\{\2n\ba{ll}
\ns\ds dX(s)=\big[ u_1(s)+ u_2(s)\big]ds+\si X(s) dW(s),\qq s\in[t,T],\\
\ns\ds X(t)=\xi,\ea\right.\ee
where $X(\cd)$ is the state process, and $u_i(\cd)$ is the control taken by Player $i$, $i=1,2$. We let
$$\cD=\{(t,\xi)\bigm|t\in[0,T),~\xi\in L^2_{\cF_t}(\O)\},$$
with
$$L^2_{\cF_t}(\Omega)=\big\{\xi:\Omega\to\dbR~| ~\xi \hbox{ is  $\cF_{t}$-measurable, } \dbE[|\xi|^2]<\i\big\}.$$
Any $(t,\xi)\in\cD$ is called an {\it initial pair}. The coefficient $\si $ is a deterministic scalar.
For $i=1,2$ and $t\in[0,T)$, the set of {\it admissible (open-loop) controls}\index{Admissible control}
of Player $i$ on $[t,T]$ is defined by
$$L^2_\dbF(t,T)=\Big\{\f:[t,T]\times\O\to\dbR\bigm|\hb{$\f(\cd)$ is $\dbF$-progressively measurable, and~}\dbE\int_t^T|\f(s)|^2ds<\i\Big\}.$$
To measure the performance  for player $1$, we introduce the functional
\bel{LQG1:cost}\ba{ll}
\ds J_1(t,\xi;u_1(\cd),u_2(\cd))=
\dbE_t\Big\{\a(T-t)  X(T)^2+\2n\int_t^T\3n\a(s-t)\(  -u_1(s)^2 +R u_2(s)^2\)ds\Big\},\ea\ee
%
where $0< R \les 1$ is a
deterministic scalars and $\a(\cd)$ is the discount function. Similarly, the performance  for player $2$ is given by
\bel{LQG2:cost}\ba{ll}
\ds J_2(t,\xi;u_1(\cd),u_2(\cd))=
\dbE_t\Big\{- \a(T-t)  X(T)^2+\2n\int_t^T\3n\a(s-t)\(  u_1(s)^2 -R u_2(s)^2\)ds\Big\}.\ea\ee
As we explain below, the condition $0<R\les 1$ means that $u_1(\cd)$ costs no less than $u_2(\cd)$ in the cost functional(s).  When $R>1$, existence of the equilibria is not guaranteed and since our objective is to maintain tractability, the assumption $0<R\leq 1$ will be maintained from now on.

 Starting from the state $\xi$,  player $i$ for $i=1,2$, selects her control $u_i(\cd)$ from the set $L^2[t,T]$ to maximize the performance $J_i(t,\xi;u_1(\cd),u_2(\cd))$. Since,
\bel{J+J=0}J_1(t,\xi;u_1(\cd),u_2(\cd))+J_2 (t,\xi;u_1(\cd),u_2(\cd))=0,\ee
our game is called a two-person zero-sum stochastic linear quadratic game. The zero sum feature capture the conflict of interest of the two players due to the presence of externalities. To discuss the implication of the model, it is useful to keep in mind a specific application of this model. Suppose that the quantity $X(T)$ represents the output of a wind turbine farm. The experienced utility of the two players is quadratic in output. The first player represents the population who value the environmental externality due to green production of electricity and experience a positive utility that is quadratic in output. The second player represents the population who experience a disutility that is also quadratic in output. The aversion to wind turbines of the second player captures opposition to wind turbines due to noise inconvenience for the citizen who live near the wind turbine farms. We assume that the level of production of the wind turbine farm is a collective decision that is taken through a political process that is mediated by political representatives. Each player, can spend some effort lobbying the politician in order to have some influence on the output decision of the turbine farm. The lobbying can take the form of campaign donations or can be more direct by sending letters to political representatives or engage in protest campaigns. The lobbying effort is captured by the control. The state equation (\ref{SDE2}) shows that positive control $u$ increases output and a negative control decreases the level of output. Player $1$ (resp. $2$) values (dislike) output and would apply $u_1= + \infty$ (resp. $u_2 = -\infty$) in the absence of constraint. However, lobbying generate a disutility that is quadratic for both players. The performance criteria \reff{LQG1:cost} shows that Player $1$ suffers from disutility from lobbying effort captured by the term $-u_1(s)^2$. The performance \reff{LQG2:cost} shows that Player $2$ suffers from disutility from lobbying effort captured by the term $-Ru_2(s)^2$ with $0<R\les1$. Therefore, for the same level of lobbying, the disutility is weakly larger for Player $1$ who supports wind turbines. This asymmetry capture the ideas that opposing an environmental reform requires less lobbying efforts because the turbines represent new technologies that will require more subsidy relative to the status quo technology for producing electricity.
Finally,  in addition to the opposite perception that the players have on wind turbine farms there is an externality in effort as well. The term $R u_2(s)^2$ in the functional  \reff{LQG1:cost} shows that more lobbying efforts from player $2$ impacts positively the performance of Player $1$. Similarly, the term $u_1(s)^2$ in the performance  \reff{LQG1:cost} shows the performance of player $2$ is impacted positively when Player $1$ exerts more lobbying efforts. This externality in lobbying efforts is tantamount to an assumption of limited supply of political capital. When player $2$ exerts more lobbying efforts against the wind turbines, that player depletes her political capital. Player $1$ is positively impacted by that depletion because Player $2$ may be less willing to lobby for other un-modelled political issues that oppose  the two players and as a result, Player $1$ may have a better outcome on those issues.

\ms

{\bf Opend-loop versus closed-loop saddle point controls}

Due to the equality \rf{J+J=0}, the performances $J_1$ and $J_2$ are not independent. If we set
$$J(t,\xi;u_1(\cd).u_2(\cd))=J_1(t,\xi;u_1(\cd),u_2(\cd)),$$
then Player $1$ wants to maximize $J(t,\xi;u_1(\cd),u_2(\cd))$ by selecting a $u_1(\cd)\in L^2_\dbF(t,T)$ and Player $2$ wants to minimize it by selecting a $u_2(\cd)\in L^2_\dbF(t,T)$. Thus, we obtain a two-person zero-sum differential game described by the state equation \rf{SDE2} and payoff/cost functional $J=J_1$ given by \rf{LQG1:cost}. For convenience, let us name it Problem (G). Consequently, in Problem (G), Player $1$ is the maximizer and Player $2$ is the minimizer. We now introduce the following definition.

\begin{defn}[\bf Open-loop saddle point] \rm  The control pair $(u^*_1(\cd), u^*_2(\cd))\in L^2_\dbF(t,T)\times L^2_\dbF(t,T)$
is called an {\it open-loop saddle point} of Problem (G) at the initial pair $(t,\xi)$, if the following holds:
\bel{definition-saddle-points}\ba{ll}
\ns\ds J(t,\xi;u_1(\cd),u_2^*(\cd))\les J(t,\xi;u^*_1(\cd),u^*_2(\cd))\les J(t,\xi;u_1^*(\cd),u_2(\cd)),\qq\forall u_1(\cd),u_2(\cd)\in L^2(t,T).\ea\ee

\end{defn}

The open-loop saddle point, if it exists, is an optimal choice for both players because if Player $i$ chooses a control different from $u^*_i(\cd)$, then her performance index would become no better or worse.  Sun and Yong (see \cite{sun2020stochastic}) showed that $(u_1^*(\cd),u_2^*(\cd))$ is an open-loop saddle point of Problem (G) on $[t,T]$ with $X^*(\cd)$ being the corresponding state process if and only if, together with another pair of adapted process $(Y^*(\cd),Z^*(\cd))$, the so-called optimality system is satisfied (with a stationary condition). This system is a system of forward-backward stochastic differential equations (FBSDEs, for short), which is actually a two-point boundary value problem of SDEs. Thus, the whole information generated by  $(X^*(s),u^*_1(s),u_2^*(s))$, $s\in[t,T]$ is needed to determine $(Y^*(\cd),Z^*(\cd))$. Whereas, the open-loop saddle point $(u_1^*(\cd),u_2^*(\cd))$ can be written in terms of $(Y^*(\cd),Z^*(\cd))$ via the stationary condition. Hence, in seeking open-loop saddle point, the information of the state, as well as the opponent's control over the whole time interval $[t,T]$ have to be used (including the initial state $\xi$). In a game situation, the future information of the state and both controls are not available at the present time to players. Therefore, open-loop saddle point only has a functional analysis meaning and it is not practically feasible. In this paper, we are primarily interested in  the so-called closed-loop saddle strategies which take the form of state-contingent affine function of the state  and is non-anticipating, meaning that future information of the state and the both controls are not needed.

We now begin with the definition of closed-loop strategies at a given time $t$. Recall that $L^2_\dbF(t,T)$ is the set of all adapted square integrable process. We also let $L^\i(t,T)$ be the set of all bounded and deterministic functions. A {\it closed-loop strategy} for both players are two couples of processes $(\Th_i(\cd),v_i(\cd))\in L^\i(t,T)\times L_\dbF^2(t,T)$, for $i=1,2$.
Under such a closed-loop strategy, the state equation reads
\bel{LQG:cl-syst}\left\{\ba{ll}
\ns\ds dX(s)=\[\(\Th_1(s)+\Th_2(s)\)X(s)+v_1(s)+v_2(s)\]ds+\si X(s) dW(s), \qq s\in[t,T],\\
\ns\ds X(t)=\xi.\ea\right.\ee
We call \reff{LQG:cl-syst} the closed-loop system under $(\Th_1(\cd),v_1(\cd);\Th_2(\cd),v_2(\cd))$
and the corresponding  performance functional reads
\bel{LQG:cost3}\ba{ll}
\ds  J(t,\xi;\Th_1(\cd)X(\cd)+v_1(\cd),\Th_2(\cd)X(\cd)+v_2(\cd))=
\dbE\Big\{\a(T-t)X(T)^2\\
\ns\ds\qq+\2n\int_t^T\3n\a(s-t)\[ -  \big (\Th_1(s)X(s)+v_1(s)\big)^2+R  \big (\Th_2(s)X(s)+v_2(s)\big)^2\]ds\Big\}.\ea\ee
To emphasize that the solution $X(\cd)$ to \rf{LQG:cl-syst} depends on $(\Th_i(\cd),v_i(\cd))$ ($i=1,2$),
as well as on the initial pair $(t,\xi)$, we frequently write
$$X(\cd)=X(\cd\,;t,\xi,\Th_1(\cd),v_1(\cd),\Th_2(\cd),v_2(\cd)).$$
The control pair $(u_1(\cd),u_2(\cd))$ defined by
\bel{outcome}u_1(\cd)=\Th_1(\cd)X(\cd)+v_1(\cd), \q u_2(\cd)=\Th_2(\cd)X(\cd)+v_2(\cd)\ee
is called the {\it outcome} of the closed-loop strategy $(\Th_1(\cd),v_1(\cd);\Th_2(\cd),v_2(\cd))$.

\ms

With the above, we are now ready to introduce the following notion.

\begin{defn}[\bf Closed-loop saddle strategies solving Problem (G)]
\rm For any initial time $ t\in[0,T)$, a closed-loop saddle strategy for Problem (G) consists of a 4-tuple $(\Th^*_1(\cd), v^*_1(\cd);\Th^*_2(\cd), v^*_2(\cd))\in
L^\i(t,T)\times L^2_\dbF(t,T)\times L^\i(t,T)\times L^2_\dbF(t,T)$ such that
for any $\xi\in L^2_{\cF_t}(\O)$, the following inequalities hold:
\bel{closed-N}\ba{ll}
\ns\ds J(t,\xi;\Th_1(\cd)X(\cd)+v_1(\cd),\Th^*_2(\cd) X(\cd)+v^*_2(\cd))\les J(t,\xi;\Th^*_1(\cd) X^*(\cd)+v^*_1(\cd),\Th^*_2(\cd) X^*(\cd)+v^*_2(\cd))\\
\ns\ds\les J(t,\xi;\Th^*_1(\cd) X(\cd)+v^*_1(\cd),\Th_2(\cd) X(\cd)+ v_2(\cd)) ,\q\forall\Th_i(\cd)\in L^\infty(t,T),~v_i(\cd)\in L^2_\dbF(t,T),~i=1,2.\ea\ee
\end{defn}

One should note that on the most left-hand sides of \rf{closed-N},
$X(\cd)=X(\cd\,;t,\xi,\Th_1(\cd),v_1(\cd),\Th^*_2(\cd),v^*_2(\cd))$,
whereas, on the most right-hand side of \rf{closed-N} $X(\cd)=X(\cd\,;t,\xi,\Th^*_1(\cd),v^*_1(\cd),\Th_2(\cd),v_2(\cd))$,
and in the middle of \rf{closed-N} $X^*(\cd)=X(\cd\,;t,\xi,\Th^*_1(\cd),v^*_1(\cd),\Th^*_2(\cd),v^*_2(\cd))$.
Thus, the processes $X(\cd)$ appearing in the most left-hand side and the most right-hand side in \rf{closed-N} are different in general.

\ms

From the above, we see that the closed-loop saddle strategy
is determined prior to the state process. In other words, the state process is determined by the state equation under the outcome \rf{outcome} of the closed-loop saddle strategy. Thus, the closed-loop saddle strategy works for arbitrary initial state $\xi$.
It follows from \rf{outcome} that the outcome, as the input of the state equation, is non-anticipating --- no future information is used. Hence, it is practically feasible. More precisely, one could calculate $(\Th_i^*(\cd),v_i^*(\cd))$ off-line first, then apply it via the outcome \rf{outcome} (the current state feedback) in the state equation.  It is crucial to underscore that this approach stands in stark contrast to the open-loop strategy (see \cite{sun2014linear}).

In this paper, we will concentrate on closed-loop strategies and their outcomes for both the single agent problem and the game problem. Note that our state equation is homogeneous and the performance index does not contain linear terms, similar to \cite{sun2020stochastic}, the control $v_1$ and $v_2$ will be zero. Thus, to simplify the notation, we will omit the $v$ component of closed-loop strategies in our subsequent analysis.

{\bf Discounting} In this paper, we consider first the standard exponential discount function $\alpha(t) = e^{\textcolor{green}{-}\rho t}$ where $\rho > 0 $ is the constant discount rate defined by $\rho \equiv - \a'(t) / \a(t)$.

To capture the present bias, we consider a mixture of exponential discount function of the form
\bel{eq: Discount function}
 \a(t) = \lambda e^{-\rho t} + (1 -\lambda) e^{-\gamma t}
\mbox{ for some } \gamma > \rho >0 \mbox{ and } \lambda \in ( 0,1).
\ee

The implied discount rate is given by
\bel{eq: Discount rate}
\ds  -\frac{\a'(t)}{\a(t)} =\rho +  \frac{ (1 -\lambda)}{ \lambda e^{(\gamma - \rho) t} +1 -\lambda} (\gamma - \rho)
\ee
and is monotonically declining from the short term discount rate $-\frac{\a'(0)}{\a(0)} = \rho+(1-\lambda)(\gamma - \rho)$ to to the long term discount rate $-\frac{\a'(\infty)}{\a(\infty)} = \rho$. Thus the discount function \reff{eq: Discount function} embodies the present bias highlighted in \cite{ainslie1975specious}. The particular form of discount function (\ref{eq: Discount function}) appears naturally for social planners in overlapping generation models (see e.g. \cite{blanchard1985debt} and \cite{calvo1988optimal}). The discount function (\ref{eq: Discount function}) was also used as a deterministic benchmark in \cite{harris2013instantaneous}.\footnote{In Section II.B of  \cite{harris2013instantaneous}, utilities at all future periods are discounted exponentially with discount factor $0 < \rho < 1$. However, distant future periods are additionally discounted with uniform weight $0 < \lambda \leq 1$. As a result, the immediate future receives full weight $ e^{ - \rho t}$, while more distant future periods are given the lower weight $\lambda e^{- \rho t}$. The immediate future lasts during a length of time that is stochastic and exponentially distributed with an intensity $\gamma - \rho$. Under these assumptions, the expected discount function matches exactly the discount function  (\ref{eq: Discount function}). }.

\section{Single Player}

In this section, we will consider the game problem with the single player. In the first part, we provide the optimal control solution when the payer has a constant discount rate. In the second part, we provide the consistent planning solution when the player has non constant discount. We will focus on the behaviour of Player $2$ because we can identify the closed loop consistent planning strategies for that players both with constant and non-constant discount rate.\footnote{The problem of player $1$ with constant discount has a solution with proper restrictions on the parameters. With non-constant discounting, to our knowledge, finding closed-loop consistent plan is an open question.}

\subsection{Single player with exponential discount}

We consider the state equation \rf{SDE2} and the objective \rf{LQG2:cost} with $u_1=0$ and $\a(t) = e^{-\rho t}$ for some $\rho>0$.
In other words, player 2 maximizes the following functional
\bel{LQC2:cost}
\ds J_2(t,\xi;0,u_2(\cd))=
\dbE\Big\{- e^{-\rho (T-t)}  X(T)^2-\2n\int_t^T\3ne^{-\rho (s-t)} R u_2(s)^2ds\Big\},
\ee
by choosing $u_2\in L^2(t,T)$ subject to
\bel{LQC2-SDE}\left\{\2n\ba{ll}
\ns\ds dX(s)= u_2(s) ds+\si X(s) dW(s),\qq s\in[t,T],\\
\ns\ds X(t)=\xi.\ea\right.\ee
This is a standard stochastic linear-quadratic optimal control problem,
which can be solved by the results presented in \cite{yong1999stochastic}.

\begin{prop}[\bf Single player with constant discount rate] \label{eq: single player optimal control} The unique optimal closed-loop optimal strategy of player 2 with state equation \rf{LQC2-SDE}
and objective \rf{LQC2:cost} is given by
\bel{eq: optimal strategy}
 \hat \Th_2(s)=-{(\rho-\si^2)  e^{(\rho-\si^2)  s}\over (1+R(\rho-\si^2) )e^{(\rho-\si^2)  T}-e^{(\rho-\si^2)  s}},\qq s\in[0,T].
\ee
In other words, for any given initial time $t\in[0,T)$ and initial state $\xi\in L^2_{\mathcal{F}_t}(\Omega;\dbR)$,
the unique optimal control is given by
\bel{eq: closed loop closed form expression}
\hat u_2(s)=\hat \Th_2(s)\hat X(s)=-{(\rho-\si^2)  e^{(\rho-\si^2)  s}\over (1+R(\rho-\si^2) )e^{(\rho-\si^2)  T}-e^{(\rho-\si^2)  s}} \hat X(s),\qq s\in[0,T],
\ee
with $\hat X$ being the unique solution of the following closed-loop system:
\bel{}\left\{\2n\ba{ll}
\ns\ds d\hat X(s)= \hat \Th_2(s) \hat X(s) ds+\si \hat X(s) dW(s),\qq s\in[t,T].\\
\ns\ds \hat X(t)=\xi.\ea\right.\ee
\end{prop}

The closed-loop optimal strategy \reff{eq: closed loop closed form expression} as a function of $s\in [0,T]$ for various levels of the model's parameters given by $(T,\sigma, \rho,R)$ is illustrated in Figure \ref{Fig: 1}. The figure indicates that the optimal feedback is negative for all parameter values, implying that lobbying against the wind turbine is optimal and lobbying effort against turbines increases with larger output since $| \hat u_2|$ increases with $\hat X$. The figure also shows that lobbying effort against the turbine increases when the cost of effort $R$ diminishes, the discount rate decreases, and the horizon $T$ decreases. When the discount rate decreases, the decision maker becomes more patient and is willing to exert more effort today to reduce output in the future. With a shorter horizon, the duration over which the decision maker can lobby is shorter, and as a result, efforts intensify. Moreover, Figure \ref{Fig: 1} demonstrates that lobbying effort intensifies when $\sigma$ increases. With more risk, it becomes more important to control the state.

\begin{figure}[h!]
  \includegraphics[width=1.0\textwidth]{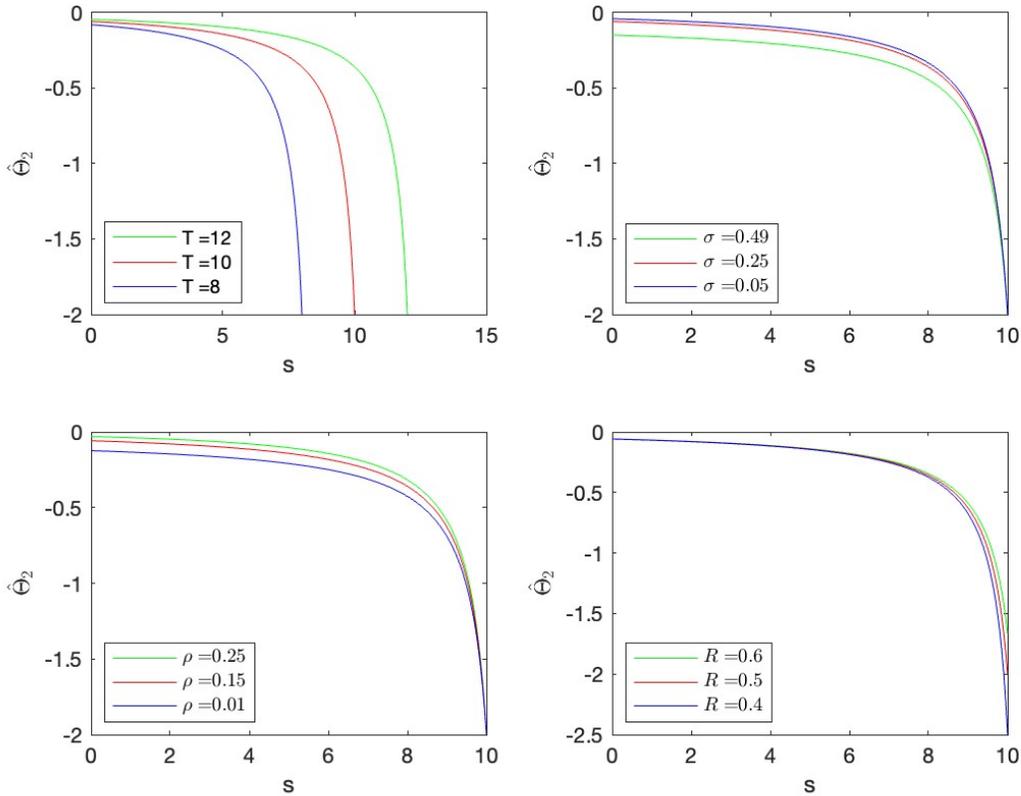}
   \caption{The panels display the optimal closed-loop control $\hat{\Theta}_2(s)$ for $0\les s \les T$ based on formula \reff{eq: optimal strategy}. The baseline parameter values used are $T = 10$, $\sigma = 0.25$, $\rho = 0.15$, and $R = 0.5$.
  }  \label{Fig: 1}
\end{figure}

\subsection{Single player with non-constant discount rate}

In this subsection, we consider the state equation \rf{SDE2} and the objective \rf{LQG2:cost} with $u_1=0$
and  the discount function $\a(t) = \lambda e^{-\rho t} + (1 -\lambda) e^{-\gamma t}$ for some $\gamma>\rho>0$,  and $\lambda\in(0,1)$.
More precisely, player 2 hopes to maximize the following functional
\bel{TI-LQC2:cost}
\ds J_2(t,\xi;0,u_2(\cd))=
\dbE\Big\{- \a(T-s) X(T)^2-\2n\int_t^T\3n \a(s-t) R u_2(s)^2ds\Big\},
\ee
by choosing $u_2\in L^2(t, T)$ subject to
\bel{TI-LQC2-SDE}\left\{\2n\ba{ll}
\ns\ds dX(s)= u_2(s) ds+\si X(s) dW(s),\qq s\in[t,T],\\
\ns\ds X(t)=\xi.\ea\right.\ee
Since the discount $\a(\cd)$ is a non-exponential function,
the above problem is time-inconsistent.
We define the consistent planning strategies that we call equilibrium closed-loop strategies as follows:

\begin{defn}[\bf Equilibrium closed-loop strategy] \rm \label{Def: CPS}
We call a closed-loop strategy $\ti \Th_2(\cd)\in L^2[0,T]$, with  $\ti X(\cd)$ being the corresponding state process,
an {\it equilibrium strategy} if
\bel{def-equilibriumC1}
\begin{aligned}
&\limsup_{\e\to 0^+} {J_2(t, X^\e(t);0,\Th^\e_2(\cd)X^\e(\cd))
-J_2(t,\ti X(t);0,\ti \Th_2(\cd)\ti X(\cd))\over\e}\les 0,
\end{aligned}
\ee
for  any $t\in[0,T)$ and $u_2(\cd)\in L^2[t,T]$, where
\bel{def-equilibriumC2}
\Th_2^\e(s)X^\e(s)\deq\left\{\2n\ba{ll}
\ds\ti\Th_2(s)X^\e(s),\q&s\in[t+\e,T];\\
\ns\ds u_2(s),\q&s\in[t,t+\e),\ea\right.\ee

with $X^\e(\cd)$ being the corresponding state process.
\end{defn}

The term ``equilibrium strategy" refers to the consistent planning approaches envisioned by \cite{strotz1956myopia}. According to the definition in condition \reff{def-equilibriumC1}-\reff{def-equilibriumC2}, an equilibrium strategy exists when no player who can control the system during the time interval $[t, t + \epsilon]$ has an incentive to deviate from the closed-loop policy $\ti {\Theta}_2$, provided that all players apply the same strategy $\ti {\Theta}_2$ during the remaining period $[t+\epsilon, T]$. This condition need to be satisfied when $\e$ is arbitrarily small. In other words, an equilibrium strategy is a closed-loop control that allows all players to maximise their respective objectives without any player having an incentive to deviate from it.  The construction \reff{def-equilibriumC1}-\reff{def-equilibriumC2} is sometimes called the spike variation approach introduced first in a deterministic setting by \cite{ekeland2006noncommitment}). Using the multi-person differential game method given by \cite{yong2017linear}, we get the corresponding equilibrium strategy.

\begin{prop}[\bf Characterization of the equilibrium closed-loop strategy] \label{Prop: closed-loop} The  equilibrium strategy associated with the  problem \reff{TI-LQC2:cost}--\reff{TI-LQC2-SDE} is given by
\bel{}
\ti \Th_2(s)=-{P(s,s)\over R},\qq s\in[0,T],
\ee
where, for any $0 \les t \les s \les T$,  the function $P(\cd,\cd)$ is the unique solution of the following Riccati equation:
\bel{RE-TIC}\left\{
\begin{aligned}
& P_s(t,s)+P(t,s)\si^2-2{ P(t,s)P(s,s)\over R}\\
&\q+[\l e^{-\rho(s-t)}+(1-\l)e^{-\gamma(s-t)}]{P(s,s)^2\over R}=0,\\
& P(t,T)=\l e^{-\rho(T-t)}+(1-\l)e^{-\g(T-t)}.
\end{aligned}\right.
\ee
\end{prop}

\begin{remark}
The existence and uniqueness of the differential equation \reff{RE-TIC} is proven in \cite{yong2017linear}. The results in Proposition~\ref{Prop: closed-loop} simply apply the result of \cite{yong2017linear} to a special case. \footnote{Theorem \ref{thn:well-ERE}, which will be presented later in this paper, also establishes the solvability of \rf{RE-TIC} but with a different method from that of \cite{yong2017linear}.} 
%
%

\end{remark}

Note that \rf{RE-TIC} is a non-local ordinary differential equation (ODE, for short),
whose solution cannot be solved explicitly. The non local feature stems from the fact that the equation  \rf{RE-TIC} involves the evaluation of the solution at two different point in time $t$ and $s$, that is, outside the diagonal of $\Delta[0,T]$.

We now describe an algorithm that will provide
an explicit approximate sequence for the solution of  \rf{RE-TIC}.

\ms


\ms

Let $\Pi$ be a partition of the time interval $[0,T]$, with $t_0=0$, $t_1={T\over N}$, $t_2={2T\over N}$,..., $t_{N-1}={(N-1)T\over N}$, $t_N=T$.
Then $\|\Pi\|=\max_{i\ges1}(t_i-t_{i-1})={1\over N}$. The algorithm begins by constructing an approximation on the final subinterval, $\Delta[t_{N-1}, t_N]$, and then proceeds to build the remainder of the approximation through backward recursion.

\ms
\noindent
\textbf{Step 1: Approximation on $\D[t_{N-1},t_N]$:} Let
\bel{Algorithm step 1}
P(t_{N-1};s)={1\over {1\over e^{\si^2 (T-s)}\a(T-t_{N-1})}+\int_s^T {1\over  e^{\si^2 (r-s)}\a(r-t_{N-1})R}dr},\qq  s\in[t_{N-1},T].
\ee
Denote
\bel{Algorithm step 2}
\begin{aligned}
&P^\Pi(t,s)=P(t_{N-1};s),\qq  (t,s)\in\D[t_{N-1},T],\\
&\Th_2^\Pi(s)=- {1\over \a(s-t_{N-1})R}P(t_{N-1};s),\qq s\in[t_{N-1},t_N].
\end{aligned}
\ee

%

\ms
\noindent
\textbf{Remaining steps: Approximation on $\D[t_{k},t_N]$ with $k=0,1,...,N-2$:} Assuming that  $P^\Pi(\cd,\cd)$ has been determined on the set $\D[t_{k+1},T]$
and  that $\Th^\Pi_2(\cd)$ has been determined on the interval  $[t_{k+1},T]$.
Let
\bel{Algorithm step 3}\begin{aligned}
&P(t_{k};s)=\a(T-t_{k})e^{\int_s^T [2\Th_2^\Pi(\t)+\si^2]d\t}+\int_s^T e^{\int_s^r [2\Th_2^\Pi(\t)+\si^2]d\t}\a(r-t_k)R\Th_2^\Pi(r)^2dr,\\\
& s\in[t_{k+1},t_N];\qq P(t_{k};s)={1\over {1\over e^{\si^2 (t_{k+1}-s)} P(t_{k};t_{k+1})}+\int_s^{t_{k+1}} {1\over e^{\si^2 ( r-s)} \a(r-t_k)R}dr},\qq s\in[t_k,t_{k+1}].
\end{aligned}
\ee
Denote
\bel{Algorithm step 4}
\begin{aligned}
&P^\Pi(t,s)=P(t_{k};s),\qq  (t,s)\in\D[t_{k},T]\setminus\D[t_{k+1},T],\\
&\Th_2^\Pi(s)=- {1\over \a(s-t_{k})R}P(t_{k};s),\qq s\in[t_k,t_{k+1}].
\end{aligned}
\ee

Note that the above approximation on $\D[t_{k},t_N]$ gives a formula of general term. Then for any $N>0$,
$P^\Pi(\cd,\cd)$ and $\Th_2^\Pi(\cd)$ can be explicitly obtained on $[0,T]$ by induction.

We now state a convergence result in Theorem \ref{thm: approx single player}. The result is due to \cite{yong2017linear} but for completeness, we give a shorter and self contained proof in the appendix.

\begin{thm}[\bf Convergence to the equilibrium closed-loop strategy]\label{thm: approx single player}
The unique solution $P(\cd,\cd)$ of \reff{RE-TIC} can be obtained as the limit
\bel{}
P(t,s)=\lim_{\|\Pi\|\to 0}P^\Pi(t,s),\qq
\ee
where the function $P^\Pi$ is produced by  the algorithm \reff{Algorithm step 1}--\reff{Algorithm step 4}.
Thus the equilibrium strategy $\ti\Th_2(\cd,\cd)$ can also be obtained as the limit
\bel{eq: approximation single player}
\ti \Th_2(s)=\lim_{\|\Pi\|\to 0}\Th^\Pi(s)=\lim_{\|\Pi\|\to 0}-{P^\Pi(s,s)\over R},\qq s\in[0,T].
\ee
\end{thm}

Figure \ref{Fig: 2} illustrates the equilibrium strategy $\ti \Th_2(\cdot)$ obtained through the use of algorithm \reff{Algorithm step 1}-\reff{Algorithm step 4}, where the step size of the partition $\Pi$ approaches zero. The results show that the optimal strategy's properties in Figure \ref{Fig: 1} hold for the equilibrium strategy as well. Additionally, the absence of "overshooting" is demonstrated in Figure \ref{Fig: 3}, where the lobbying equilibrium with non-constant discount is bounded by the optimal lobbying strategy, both when the short-term self is in control and when the long-term self is in control.

\begin{figure}[h!]
  \includegraphics[width=0.8\textwidth]{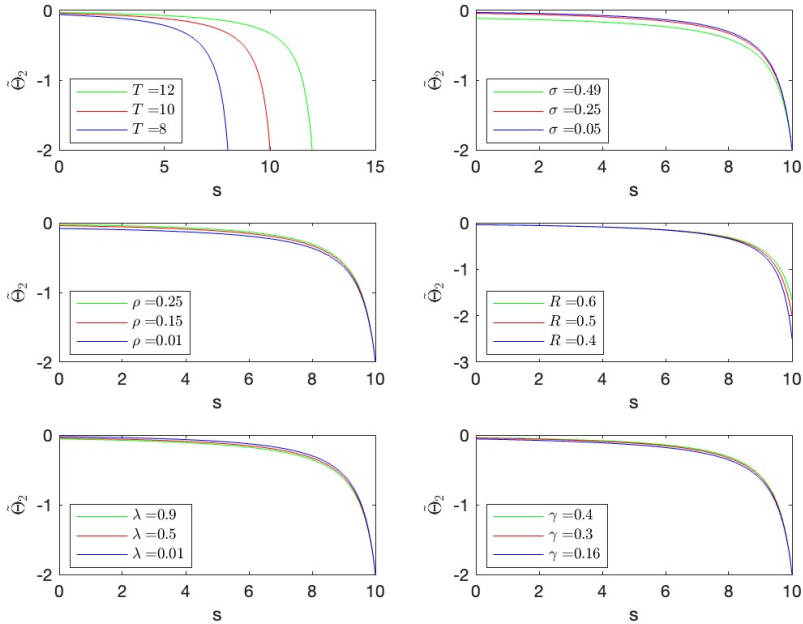}
   \caption{The panels display the equilibrium closed-loop control $\tilde{\Theta}_2(s)$ for $0 \leq s \leq T$ associated with the single player game with non-constant discounting, and is based on the approximation given in equation \reff{eq: approximation single player}. The baseline parameter values used are $T = 10$, $\sigma = 0.25$, $\rho = 0.15$, $R = 0.5$, $\lambda = 0.3$, and $\gamma = 0.3$.
  } \label{Fig: 2}
\end{figure}

\begin{figure}[h!]
  \includegraphics[width=0.85\textwidth]{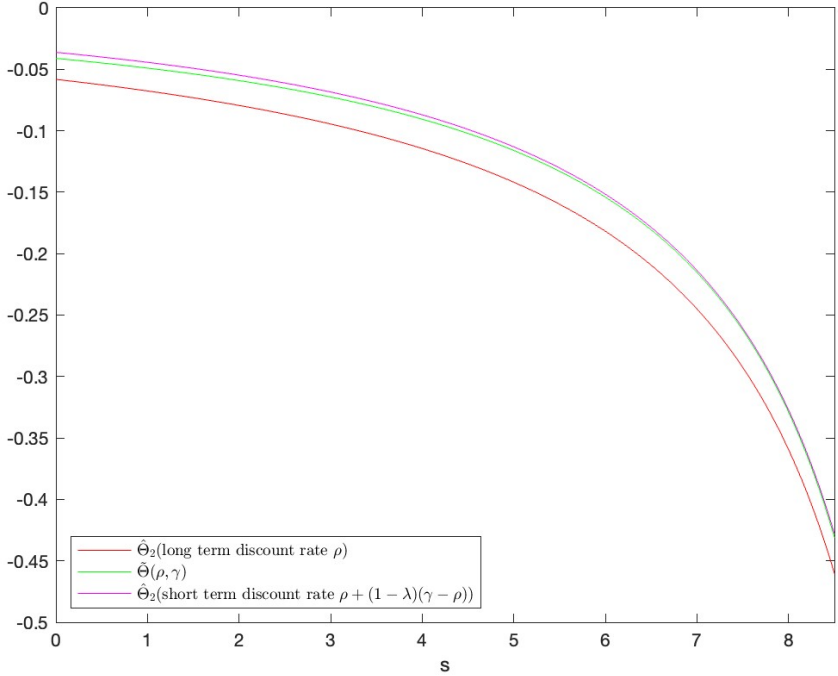}
  \caption{The figure shows the equilibrium closed-loop control $\tilde{\Theta}_2(s)$ in green, and the optimal closed-loop control $\hat{\Theta}_2(s)$ when the discount rate is $\rho$ (in pink) or $\rho + (1-\lambda)(\gamma -\rho)$ (in red) for $0 \les s \les T$. The strategy $\tilde{\Theta}_2(s)$ is obtained using the approximation \reff{eq: approximation single player}. The baseline parameter values are $T = 10$, $\sigma = 0.25$, $\rho = 0.15$, $R = 0.5$, $\lambda = 0.5$, and $\gamma = 0.3$.
  } \label{Fig: 3}
\end{figure}

\newpage

\section{Two-Person Zero-Sum Games with Constant Discounting}

In this section, we shall consider the game problem
described by \rf{SDE2}--\rf{LQG1:cost}--\rf{LQG2:cost} with $\a(t) = e^{-\rho t}$ for some $\rho>0$. Recall that we denoted the game problem by Problem (G). We first define a more general auxiliary game that will be used for solving our game both with constant and non constant discounting. Second, we provide the closed-form closed-loop saddle strategy for Problem (G) with constant discounting.

\ms
\subsection{An auxiliary game}
We introduce the following more general game problem with the state \rf{SDE2} and the objective
\bel{LQG1:cost1}\ba{ll}
\ds J(t,\xi;u_1(\cd),u_2(\cd))=
\dbE_t\Big\{e^{-\rho (T-t)} G X(T)^2+\2n\int_t^T\3n e^{-\rho(s-t)}\( R_1(s) u_1(s)^2 +R_2(s) u_2(s)^2\)ds\Big\},\ea\ee
%
%
%
%
for some scalar $G$ and deterministic functions $R_1$ and $R_2$. Note that the above  reduces to Problem (G) with constant discounting
when $G=1$, $R_1(\cd)=-1$ and $R_2(\cd)=R$.
By  \cite{sun2014linear} (Theorem $5.2$ in page $4103$), we have the following characterization of the closed-loop saddle strategy of the above game .

\begin{lem}[\bf Auxiliary result]\label{lem:LQG-closed-Nash-kehua}
Suppose that $G \ges 0$, $R_1(\cd)< 0$ and $R_2(\cd)> 0$.
The game problem with state \rf{SDE2} and objectives \rf{LQG1:cost1} admits a closed-loop saddle strategy
$(\Th^*_1(\cd),$ $\Th^*_2(\cd))$ if and only if the Riccati equation
\bel{Ric}\left\{\begin{aligned}
  & \dot{P}(s)+\si^2 P(s)-\rho P(s) -{R_1(s)+R_2(s)\over R_1(s)R_2(s)}P(s)^2=0, \\
  & P(T)=G,
\end{aligned}\right.\ee
admits a solution $P(\cd)\in C([t,T];\dbR)$.
In this case, the unique closed-loop Nash strategy $(\Th^*_1(\cd), \Th^*_2(\cd))$
admits the representation:
\bel{Th^*}
  \Th^*_1(s)=-{P(s)\over R_1(s)},\qq \Th^*_2(s)=-{P(s)\over R_2(s)},\qq s\in[0,T].
\ee
\end{lem}

\subsection{Closed-loop saddle strategies}

We now turn to Problem (G) with constant discounting. By Lemma \ref{lem:LQG-closed-Nash-kehua}, we get the following result.

\begin{prop}{\bf Closed-loop saddle Strategies for Problem (G) with constant discounting.}\label{thm:LQG-closed-Nash-kehua}
 \rm Problem ((G) with $\alpha(t) = e^{- \rho t}$ admits a unique closed-loop saddle strategy
$(\Th^*_1(\cd),\Th^*_2(\cd))$, given by
\bel{eq: Nash for Problem N}\begin{aligned}
\Th_1^*(s)&={R(\rho-\si^2) e^{(\rho-\si^2) s}\over [1-R+R(\rho-\si^2) ]e^{(\rho-\si^2)  T}-(1-R)e^{(\rho-\si^2)  s}},\\
\Th_2^*(s)&=-{(\rho-\si^2) e^{(\rho-\si^2) s}\over [1-R+R(\rho-\si^2) ]e^{(\rho-\si^2)  T}-(1-R)e^{(\rho-\si^2)  s}},
\end{aligned}\qq s\in[0,T].
\ee
In particular, if $R=1$, under which the game is symmetric,
then the unique closed-loop Nash strategy $(\Th^*_1(\cd), \Th^*_2(\cd))$
admits the representation:
\bel{}
\Th_1^*(s)=(\rho-\si^2) e^{(\rho-\si^2) (s-T)},\qq
\Th_2^*(s)=-(\rho-\si^2) e^{(\rho-\si^2) (s-T)},\qq s\in[0,T].
\ee
\end{prop}

To visualize the saddle interaction between the two players, Figure \ref{Fig: 4} plots the closed-loop Nash strategies for Problem (G) as a function of time, at various parameter levels for player $2$. We observe that many of the comparative static intuitions from the single player game hold true in the two player game, with one notable exception: the impact of the cost of lobbying parameter $R$. As shown in the east-south panel of Figure \ref{Fig: 4}, lobbying efforts may be smaller when the cost of lobbying is lower. However, when sufficiently close to the terminal date, the order reverts back to the configuration seen in the single player game.

\begin{figure}[h!]
  \includegraphics[width=1.0\textwidth]{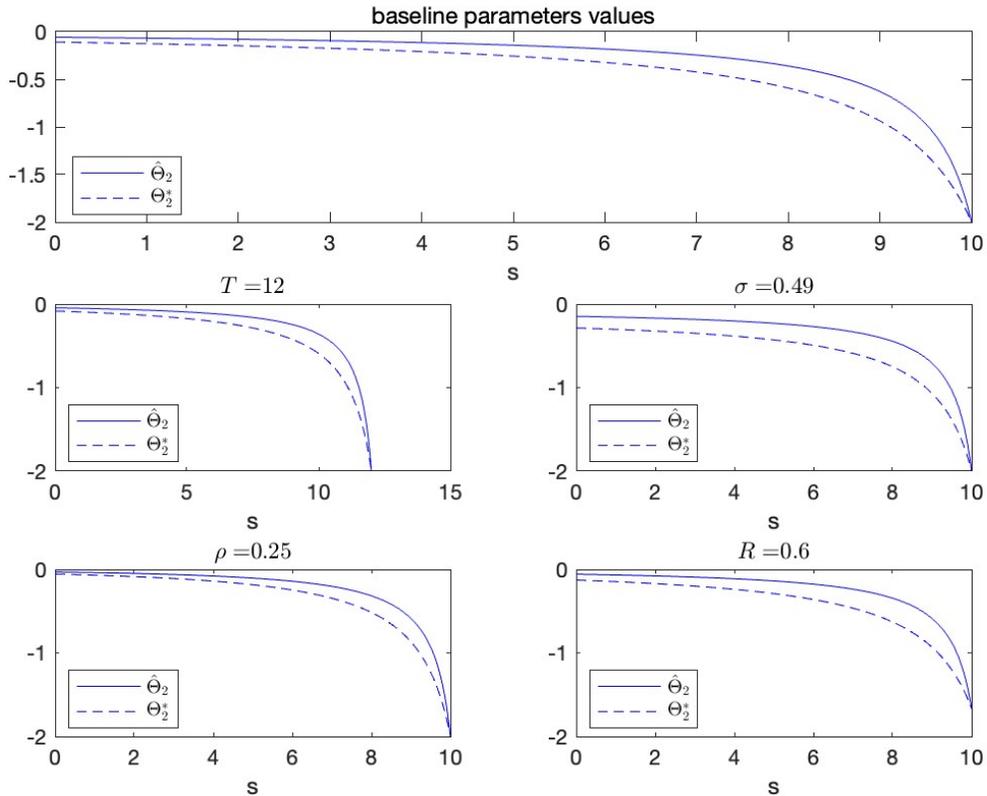}
   \caption{The figure panels display the closed-loop saddle strategies $\Theta_2^*(s)$ for Problem (G) with constant discounting, plotted for $0 \les s \les T$, using the closed-form expression given by equation \reff{eq: Nash for Problem N}. We use the baseline parameter values of $T = 10$, $\sigma = 0.25$, $\rho = 0.15$, $R = 0.5$.
  } \label{Fig: 4}
\end{figure}

Figure \ref{Fig: 5} presents a comparison between Player 2's closed-loop saddle lobbying strategy and the optimal lobbying strategy for a single player. The figure shows that lobbying becomes more intense in the two-player game than in the single player optimization. This trend is consistent across all panels in Figure \ref{Fig: 5} and highlights how strategic behavior translates into preemptive lobbying.

\begin{figure}[h!]
  \includegraphics[width=1.0\textwidth]{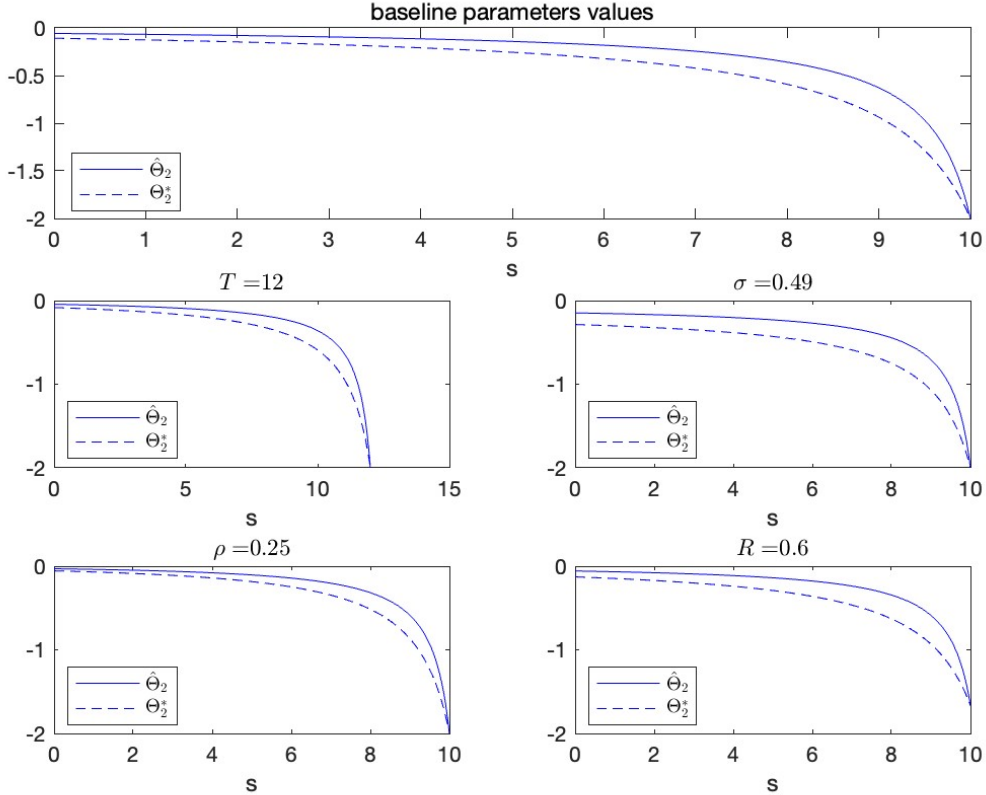}
   \caption{The figure panels display the closed-loop saddle strategies $\Theta_2^*(s)$ for Problem (G) with constant discounting, plotted for $0 \les s \les T$, along with the optimal closed-loop strategy $\hat{\Theta}_2(s)$. We use the baseline parameter values of $T = 10$, $\sigma = 0.25$, $\rho = 0.15$, $R = 0.5$, $\lambda = 0.5$, and $\gamma = 0.3$.
  } \label{Fig: 5}
\end{figure}

\section{Two-Person Zero-Sum Game with Non-constant Discounting} \label{section: zero sum with non constant discount}

In this section, we explore Problem (G) in the case where the discount rate is not constant, and the discount function is defined by \reff{eq: Discount function}. First, we define the equilibrium. Second, we introduce a modified version of Problem (G) in which we divide the time interval $[0,T]$ into subintervals, and assume that players can make commitments during each subinterval. This discretization enables us to identify a differential equation that has the potential to characterize the equilibrium. Third, we demonstrate that the problem is well-posed and that the equilibrium can indeed be characterized by the identified equation. Finally, we present an algorithm that can be used to approximate the equilibrium.

\subsection{Equilibrium definition}

The introduction of non-constant discounting in Problem (G) adds a strategic dimension for each player in two ways. Firstly, each player is strategic in their interaction with the other player. Secondly, each player is strategic in their interaction with their future selves. Consequently, Problem (G) effectively becomes a problem with an infinite number of players: the continuum of incarnations of Player $1$ and the continuum of incarnations of Player $2$. To account for the time inconsistency generated by non-constant discounting, we define the equilibrium closed-loop Nash strategies as follows.

\begin{defn}{\bf Equilibrium closed-loop saddle strategies.}\rm
We say a closed-loop strategy $\bar\Th_i(\cd)\in L^2[0,T]$ ($i=1,2$), with  $\bar X(\cd)$ being the corresponding state process,
satisfies a {\it local saddle property} if
\bel{def-equilibrium1}
\begin{aligned}
&\limsup_{\e\to 0^+} {J(t, X^e(t);\Th_1^\e(\cd)X^\e(\cd),\bar\Th_2(\cd)X^\e(\cd))
-J(t,\bar X(t);\bar\Th_1(\cd)\bar X(\cd),\bar\Th_2(\cd)\bar X(\cd))\over\e}\les 0,\\
&\limsup_{\e\to 0^+} {J(t, X^\e(t);\bar\Th_1(\cd)X^\e(\cd),\Th^\e_2(\cd)X^\e(\cd))
-J(t,\bar X(t);\bar\Th_1(\cd)\bar X(\cd),\bar\Th_2(\cd)\bar X(\cd))\over\e}\ges 0,
\end{aligned}
\ee
for  any $t\in[0,T)$ and $u_i(\cd)\in L^2_i[t,T]$ with $i=1,2$, where
\bel{def-equilibrium2}
\Th_i^\e(s)X^\e(s)\deq\left\{\2n\ba{ll}
\ds\bar\Th_i(s)X^\e(s),\q&s\in[t+\e,T];\\
\ns\ds u_i(s),\q&s\in[t,t+\e),\ea\right.\ee
with $X^\e(\cd)$ being the corresponding state process.
We call the closed-loop strategies satisfying the local Nash property \rf{def-equilibrium1}
an {\it equilibrium closed-loop saddle strategy}.
\end{defn}

 The first inequality in \reff{def-equilibrium1} implies that if player $1$ deviates from their strategy during a commitment period of negligible length, their objective function will not improve compared to following the no-deviation strategy. This constraint shares similarities with the definition of a consistent planning strategy for a single player, as stated in Definition \ref{Def: CPS}. However, the key difference is that in the game-theoretic setting, each player takes turns to deviate while keeping the behavior of both their future selves during the period $[t + \e, T]$ and the other player's strategy unchanged. Strategies that meet the criteria of the local saddle property \reff{def-equilibrium2} are denominated as equilibrium closed-loop saddle strategies. Here, ``equilibrium" pertains to the game involving multiple selves, while ``saddle" refers to the game between Player $1$ and Player $2$.

\subsection{Time partition with precommitment in each subinterval}
We first let $\Pi$ be a partition of the time interval $[0,T]$:
$$0=t_0<t_1<t_2<\cds<t_{N-1}<t_N=T,$$
with mesh size
$$\|\Pi\|=\max_{i\ges1}(t_i-t_{i-1}).$$
We assume that within each subinterval $[t_k, t_{k+1}]$, the players are allowed to make commitments and adjust their strategies accordingly, before moving on to the next subinterval. As a result, during the time interval $[t_k, t_{k+1}]$, self $t_k$ of each player will play a 2-person zero-sum game against self $t_k$ of the other player, with self $t_k$ dictating the strategy choice. To fully specify the problem, we need to define the objective for each player in each subinterval, and establish the connection between the game in a given subinterval and the game in the next subinterval. To establish this connection, we will assume consistent planning: when solving the game during a given interval $[t_{k-1}, t_{k}]$, the decision makers internalize how the game will be solved in the subsequent interval $[t_{k},T]$.  We will do this in several steps.

\ms

\noindent
\textbf{Step $1$: The precommitment game on $\pmb{[t_{N-1},t_N]}$.}
The precommitment game that we envision is driven the state equation
\bel{SDE-N}\left\{\2n\ba{ll}
\ns\ds dX^N(s)=\big[ u^N_1(s)+ u^N_2(s)\big]ds+\si X^N(s) dW(s),\qq s\in[t_{N-1},T],\\
\ns\ds X^N(t_{N-1})=\xi,\ea\right.\ee
and the functionals
\bel{cost-N}\ba{ll}
\ds J^N_1(t_{N-1},\xi;u_1(\cd),u_2(\cd))=\dbE\Big\{\a(T-t_{N-1})  X(T)^2\\
\ds\qq+\2n\int_t^T\3n\a(s-t_{N-1})\(  -u_1(s)^2 +R u_2(s)^2\)ds\Big\},\\
\ds J^N_2(t_{N-1},\xi;u_1(\cd),u_2(\cd))=
\dbE\Big\{- \a(T-t_{N-1})  X(T)^2\\
\ds\qq+\2n\int_t^T\3n\a(s-t_{N-1})\(  u_1(s)^2 -R u_2(s)^2\)ds\Big\}.\ea\ee
Notice that with precommitment, every player's future self $t$ where $t \in [t_{N-1},t_N]$ apply the same discount function that is applied by self $t_{N-1}$. As a result,  the initial time $t_{N-1}$ is fixed in equations \reff{cost-N}. With fixed $t_{N-1}$,   the problem \reff{SDE-N}--\reff{cost-N} is a standard linear quadratic game with zero discounting. As such, it can be solved Lemma \ref{lem:LQG-closed-Nash-kehua}, with $\rho=0$,
$G=\a(T-t_{N-1})$, $R_1(s)=-\a(s-t_{N-1})$ and $R_2(s)=\a(s-t_{N-1})R$.

The unique closed-loop saddle strategy $(\bar\Th_1^N(\cd), \bar\Th_2^N(\cd))$ is given by
\bel{}
\begin{aligned}
  \bar\Th^N_1(s)={ P(t_{N-1};s)\over \a(s-t_{N-1})},\qq \bar\Th^N_2(s)=-{ P(t_{N-1};s)\over R\a(s-t_{N-1})},\qq s\in[t_{N-1},T],
\end{aligned}
\ee
with $P(t_{N-1};\cd)$ being the unique solution of the following Riccati equation:
\bel{Ric-N}\left\{\begin{aligned}
  & P_s(t_{N-1};s)+ P(t_{N-1};s)\si^2 - {1-R\over \a(s-t_{N-1})R}P(t_{N-1};s)^2=0,
  \qq s\in[t_{N-1},T], \\
  & P(t_{N-1};T)=\a(T-t_{N-1}).
\end{aligned}\right.\ee
Note that in \rf{Ric-N}, $t_{N-1}$ is fixed and only works as a parameter.
Denote
\bel{Th-Pi-N}
  \bar\Th^\Pi_i(\cd)\deq\bar\Th^N_i(\cd) \hbox{ on $[t_{N-1},T]$}, \q i=1,2.
\ee

\noindent
\textbf{Steps $2$: The precommitment game on $\pmb{[t_{N-2},t_{N-1}]}$.}
With the strategy $\bar\Th^\Pi_i(\cd)$ determined on $[t_{N-1},T]$, the players on $[t_{N-2},T]$ can only control the system on $[t_{N-2},t_{N-1}]$.
Consider the state equation
\bel{SDE-(N-1)}\left\{\2n\ba{ll}
\ns\ds dX^{N-1}(s)=\big[u^{N-1}_1(s)+ u^{N-1}_2(s)\big]ds+\si X^{N-1}(s) dW(s),\qq s\in[t_{N-2},t_{N-1}],\\
\ns\ds dX^{N-1}(s)=\big[ \bar\Th_1^\Pi(s) X^{N-1}(s)+\bar\Th_2^\Pi(s) X^{N-1}(s)\big]ds+\si X^{N-1}(s) dW(s),\q s\in[t_{N-1},T],\\
\ns\ds X^{N-1}(t_{N-2})=\xi,\ea\right.\ee
and the functionals $J_i^{N-1}(t_{N-2},\xi;u^{N-1}_1(\cd),u^{N-1}_2(\cd)),i=1,2$ with
\bel{cost-(N-1)}\ba{ll}
\ns\ds J_1^{N-1}(t_{N-2},\xi;u^{N-1}_1(\cd),u^{N-1}_2(\cd))=- J_2^{N-1}(t_{N-2},\xi;u^{N-1}_1(\cd),u^{N-1}_2(\cd)) \\
\ns\ds \q=J^{N-1}(t_{N-2},\xi;u^{N-1}_1(\cd),u^{N-1}_2(\cd))=\dbE\Big\{\a(T-t_{N-2})  X^{N-1}(T)^2\\
\ns\ds\qq+\2n\int_{t_{N-1}}^T\3n\a(s-t_{N-2})\[-\big(\bar\Th_i^\Pi(s) X^{N-1}(s)\big)^2+R\big(\bar\Th_i^\Pi(s) X^{N-1}(s)\big)^2\]ds\Big\}\\
\ns\ds\qq+\dbE\Big\{\int_{t_{N-2}}^{t_{N-1}}\3n\a(s-t_{N-2})\[-u^{N-1}(s)^2+ R u_2(s)^2\]ds\Big\}\deq (I)+(II).
\ea
\ee
Let $P(t_{N-2};\cd)$ be the unique solution of the following Lyapunov equation over $[t_{N-1},T]$:
\bel{Lya-(N-1)}\left\{\begin{aligned}
  & P_s(t_{N-2};s)+2 P(t_{N-2};s)[\bar\Th^\Pi_1(s)+\bar\Th^\Pi_2(s)]+P(t_{N-2};s)\si^2 \\
  &\q+\a(s-t_{N-2})\big[-\bar\Th_1^\Pi(s)^2+R\bar\Th_2^\Pi(s)^2\big]=0, \\
  & P(t_{N-2};T)=\a(T-t_{N-2}).
\end{aligned}\right.\ee
Then by applying It\^{o} formula to $P(t_{N-2};\cd)X^{N-1}(\cd)^2$ on $[t_{N-1},T]$, we have
$$
(I)=P(t_{N-2};t_{N-1})X^{N-1}(t_{N-1})^2.
$$
It follows that
\bel{cost-(N-1)-1}\ba{ll}
\ns\ds J^{N-1}(t_{N-2},\xi;u^{N-1}_1(\cd),u^{N-1}_2(\cd))=\dbE\Big\{P(t_{N-2};t_{N-1})X^{N-1}(t_{N-1})^2\\
\ns\ds\qq+\int_{t_{N-2}}^{t_{N-1}}\3n\a(s-t_{N-2})\[-u^{N-1}(s)^2+ R u_2(s)^2\]ds\Big\}.
\ea
\ee
Moreover, from the fact
$$
-\bar\Th_1^\Pi(s)^2+R\bar\Th_2^\Pi(s)^2
={(1-R) P(t_{N-1};s)^2\over R\a(s-t_{N-1})^2}\ges 0,
$$
we have
\bel{P-(N-1)}
P(t_{N-2};t_{N-1})\ges0.
\ee
Thus, the problem can be solved  by Lemma \ref{lem:LQG-closed-Nash-kehua} again with  with $\rho=0$,
$G=P(t_{N-2};t_{N-1})$, $R_1(s)=-\a(s-t_{N-2})$ and $R_2(s)=\a(s-t_{N-2})R$.

The unique closed-loop saddle strategy $(\bar\Th_1^{N-1}(\cd), \bar\Th_2^{N-1}(\cd))$ is given by
\bel{Th-Pi-N-1-i}
\begin{aligned}
  \bar\Th^{N-1}_1(s)={ P(t_{N-2};s)\over \a(s-t_{N-2})},\qq \bar\Th^{N-1}_2(s)=-{ P(t_{N-2};s)\over R\a(s-t_{N-2})},\qq s\in[t_{N-2},t_{N-1}],
\end{aligned}
\ee
with $P(t_{N-2};\cd)$ being the unique solution of the following Riccati equation:
\bel{Ric-(N-1)}\left\{\begin{aligned}
   & P_s(t_{N-2};s)+ P(t_{N-2};s)\si^2 - {1-R\over \a(s-t_{N-2})R}P(t_{N-2};s)^2=0,
  \qq s\in[t_{N-2},t_{N-1}],\\
  & P(t_{N-2};t_{N-1})=\ P(t_{N-2};t_{N-1}).
\end{aligned}\right.\ee
Note that in the above,  $P(t_{N-2};t_{N-1})$ has been determined by \rf{Lya-(N-1)}.
We now extend $\bar\Th^\Pi(\cd)$ from $[t_{N-1},T]$ to $[t_{N-2},T]$ by
\bel{Th-Pi-N-1}
  \bar\Th^\Pi_i(s)\deq\left\{\begin{aligned}
  & \bar\Th^N_i(s),\qq s\in[t_{N-1},T];\\
  & \bar\Th^{N-1}_i(s),\qq s\in[t_{N-2},t_{N-1}),
\end{aligned}\right.\qq i=1,2.
\ee
With this, we can write \rf{Lya-(N-1)} and \rf{Th-Pi-N-1} together as
\bel{ERIC-(N-1)}\left\{\begin{aligned}
  & P_s(t_{N-2};s)+2 P(t_{N-2};s)[\bar\Th^\Pi_1(s)+\bar\Th^\Pi_2(s)]+ P(t_{N-2};s)\si^2 \\
  &\q+\a(s-t_{N-2})\big[-\bar\Th_1^\Pi(s)^2+R\bar\Th_2^\Pi(s)^2\big]=0, \qq s\in[t_{N-2},T],\\
  & P(t_{N-2};T)=\a(T-t_{N-2}).
\end{aligned}\right.\ee

\ms

\noindent
\textbf{Subsequent Steps: The precommitment game on $\pmb{[t_{k-1},t_{k}]}$ for any  $\pmb{k=1,2,...,N}$.}
Suppose $\bar\Th^\Pi_i(s),i=1,2$ has been constructed on $[t_k,T]\times\dbR^n$, (for some $k=1,2,\cds,N-1$). We apply the above strategy extension procedure to obtain an extension $\bar\Th^\Pi_i(\cd):[t_{k-1},T]\to\dbR$ of  $\bar\Th^\Pi_i(s),s\in[t_{k},T]$ by the following steps:
\bel{Th-Pi-N-1}
  \bar\Th^\Pi_i(s)\deq\left\{\begin{aligned}
  &\bar\Th^\Pi_i(s),\qq s\in[t_k,T];\\
  & \bar\Th^{k}_i(s),\qq s\in[t_{k-1},t_{k}),
\end{aligned}\right.\qq i=1,2.
\ee
where
\bel{}
\begin{aligned}
  \bar\Th^{k}_1(s)={ P(t_{k-1};s)\over \a(s-t_{k-1})},\qq \bar\Th^{k}_2(s)=-{ P(t_{k-1};s)\over R\a(s-t_{k-1})},\qq s\in[t_{k-1},t_{k}],
\end{aligned}
\ee
with $P(t_{k-1};\cd)$ being the unique solution of the following Riccati equation:
\bel{Lya-k}\left\{\begin{aligned}
  & P_s(t_{k-1};s)+2 P(t_{k-1};s)[\bar\Th^\Pi_1(s)+\bar\Th^\Pi_2(s)]+P(t_{k-1};s)\si^2, \\
  &\q+\a(s-t_{k-1})\big[-\bar\Th_1^\Pi(s)^2+R\bar\Th_2^\Pi(s)^2\big]=0, \qq s\in[t_k,T];\\
   & P_s(t_{k-1};s)+ P(t_{k-1};s)\si^2 - {1-R\over \a(s-t_{k-1})R}P(t_{k-1};s)^2=0,
  \qq s\in[t_{k-1},t_{k}],\\
  & P(t_{k-1};T)=\a(T-t_{k-1}),
\end{aligned}\right.\ee
which is equivalent to
\bel{Lya-k1}\left\{\begin{aligned}
  & P_s(t_{k-1};s)+2 P(t_{k-1};s)[\bar\Th^\Pi_1(s)+\bar\Th^\Pi_2(s)]+P(t_{k-1};s)\si^2, \\
  &\q+\a(s-t_{k-1})\big[-\bar\Th_1^\Pi(s)^2+R\bar\Th_2^\Pi(s)^2\big]=0, \qq s\in[t_{k-1},t_{k}],\\
  & P(t_{k-1};T)=\a(T-t_{k-1}).
\end{aligned}\right.\ee

\ms

This completes the induction.

\ms

The novelty of the construction of a sequence of precomitment zero sum games, is that we can show under the assumption $0<R\les 1$, the zero-sum stochastic linear quadratic game on every subinterval can be solved. The key step is then to check whether we can connect the Riccati equations associated with each subinterval's two persons zero sum game with constant discount, and then glue them together. This is what we did in the above steps.

\ms

We now describe the global policy constructed from the above steps. Denote the discount function of the precommitment model by
$$
\a^\Pi(t,s)=\left\{ \begin{aligned}
&\a(s-t_{N-1}),\qq t\in[t_{N-1},t_N],\q s\in[t,T],\\
&\a(s-t_{k-1}),\qq t\in[t_{k-1},t_k),\q s\in[t,T],\q k=1,...,N-1.
\end{aligned}\right.
$$
The above discount function formalises the assumption that self $k-1$ applies her discount function to all selves in the interval $[t_{k-1}, t_k]$. We denote by $P^{\Pi}$ the function obtained by gluing the functions $P$ on each subinterval:
$$
P^\Pi(t,s)= P(s-t_{k-1}),\qq t\in[t_{k-1},t_k),\q s\in[t,T],\q k=1,...,N.
$$
Then the closed-loop strategy $(\bar\Th_1^\Pi(\cd),\bar\Th_2^\Pi(\cd))$
associated with the partition $\Pi$ can be given by
\bel{ERIC-Theta}
  \bar\Th^\Pi_1(s)= { P^\Pi(s;s)\over \a^\Pi(s,s)},\qq \bar\Th^\Pi_2(s)=-{ P^\Pi(s;s)\over R\a^\Pi(s,s)},
\qq s\in[0,T],
\ee
with
\bel{ERIC-Pi}\left\{\begin{aligned}
  & P^\Pi_s(t,s)+2 P^\Pi(t,s)[\bar\Th^\Pi_1(s)+\bar\Th^\Pi_2(s)] +P^\Pi(t,s)\si^2 \\
  &\q+\a^\Pi(t,s)\big[-\bar\Th_1^\Pi(s)^2+R\bar\Th_2^\Pi(s)^2\big]=0, \qq (t,s)\in\D[0,T],\\
  & P^\Pi(t,T)=\a^\Pi(T-t),\qq t\in[0,T].
\end{aligned}\right.\ee

\ms

The purpose of our precommitment strategy based on subintervals constructed in the above steps is to find an equilibrium characterization in the continuous time model. To accomplish this, we pretend that $P^\Pi(\cd,\cd)$ uniformly converges to some $P(\cd,\cd)$.
Then, formally, $P(\cd,\cd)$ should satisfy the following equation:
\bel{ERIC}\left\{\begin{aligned}
& P_s(t,s)+2 P(t,s)[\bar\Th_1(s)+\bar\Th_2(s)] +P(t,s)\si^2 \\
&\q+\a(t,s)\big[-\bar\Th_1(s)^2+R\bar\Th_2(s)^2\big]=0,\qq (t,s)\in\D[0,T],\\
&P(t,T)=\a(T-t),\qq t\in[0,T].
\end{aligned}\right.\ee
where
\bel{Th}
\bar\Th_1(s)= P(s,s),\qq \bar\Th_2(s)=-{P(s,s)\over R},\qq s\in[0,T].
\ee
We call \rf{ERIC} an equilibrium Riccati equation (ERE, for short).

\subsection{Well-posedness and verification theorem}

We show in the next theorem that the solution to the ERE \rf{ERIC} exists and is unique.

\begin{thm}{\bf Existence and uniqueness of the ERE.}\label{thn:well-ERE}
The ERE \rf{ERIC} admits a unique solution $P(\cd,\cd)\in C(\D[0,T])$.
Moreover, $P(t,s)\ges 0,\, (t,s)\in\D[0,T]$. In particular, if $R=1$, that is the game is symmetric,
then the unique solution of ERE \rf{ERIC} can be explicitly given by
\bel{thn:well-ERE1}
P(t,s)=e^{-\si^2(s-T)}\a(T-t),  \qq (t,s)\in\D[t,T].
\ee
\end{thm}
Next, we shall show that the closed-loop strategy $\bar \Th_i(\cd),i=1,2$ obtained by
\rf{ERIC}--\rf{Th} satisfies the  local Nash property and is therefore an equilibrium closed-loop saddle strategy.

\begin{thm}{\bf Verification.} \label{thm:VT}
Let  $\bar \Th_i(\cd),i=1,2$ be the closed-loop strategy obtained by
\rf{ERIC}--\rf{Th}. Then $\bar \Th_i(\cd),i=1,2$ satisfies the local saddle property. Thus  $\bar \Th_i(\cd),i=1,2$ is an equilibrium closed-loop saddle strategy for Problem (G) with non constant discounting.
\end{thm}

\subsection{Convergence and approximation algorithm}

We now show that when the partition step is small, the precommitment strategy converges to the equilibrium closed-loop saddle strategy. This is an important result because it builds the equilibrium as a limit of models with precommitment over arbitrarily smaller periods. The result provides thus a discrete time foundation of the spike variation method.

\begin{thm}\label{thm:convergence}
As $\|\Pi\|\to 0$, the sequences $\{P^\Pi(\cd,\cd)\}_{\Pi}$
and $\{\bar\Th_i^\Pi(\cd,\cd)\}_{\Pi}$ defined by \reff{ERIC-Theta} and \reff{ERIC-Pi} uniformly converge to $P(\cd,\cd)$ and $\bar\Th(\cd)$, respectively where $P(\cd,\cd)$ and $\bar\Th(\cd)$ are defined by \reff{ERIC} and \reff{Th}.
\end{thm}

It should be noted that the explicit solution to the non-local ODE \rf{ERIC} cannot be obtained in general. However, for a given partition $\Pi$, we can obtain an explicit solution for $P^\Pi(\cd,\cd)$. By applying Theorem \ref{thm:convergence}, we can use ${P^\Pi(\cd,\cd)}$ to approximate the solution of \rf{ERIC}. This will be our approach going forward.

 \ms

To start, let $\Pi$ be a partition of the time interval $[0,T]$, with $t_0=0$, $t_1={T\over N}$, $t_2={2T\over N}$,..., $t_{N-1}={(N-1)T\over N}$, $t_N=T$.
Then $\|\Pi\|=\max_{i\ges1}(t_i-t_{i-1})={1\over N}$. We now describe the algorithm.

\ms
\noindent
\textbf{Step $1$: Approximation on $\D[t_{N-1},t_N]$:} Let
\bel{}
P(t_{N-1};s)={1\over {1\over e^{\si^2 (T-s)}\a(T-t_{N-1})}+\int_s^T {1-R\over  e^{\si^2( r-s)}\a(r-t_{N-1})R}dr},\qq  s\in[t_{N-1},T].
\ee
Denote
\bel{}
\begin{aligned}
&P^\Pi(t,s)=P(t_{N-1};s),\qq  (t,s)\in\D[t_{N-1},T],\\
&\Th_1^\Pi(s)={ P(t_{N-1};s)\over \a(s-t_{N-1})},\qq \Th^\Pi_2(s)=-{ P(t_{N-1};s)\over R\a(s-t_{N-1})},\qq s\in[t_{N-1},t_N].
\end{aligned}
\ee

\ms
\noindent
\textbf{Step $2$: Approximation formula of general term on $\D[t_{k},t_N]$ with $k=0,1,...,N-2$:} Let $P^\Pi(\cd,\cd)$ have been determined on $\D[t_{k+1},T]$
and  $\Th^\Pi(\cd)$ have been determined on $[t_{k+1},T]$.
Let
\bel{}\begin{aligned}
&P(t_{k};s)=\a(T-t_{k})e^{\int_s^T [2\Th^\Pi_1(\t)+2 \Th^\Pi_2(\t)+\si^2]d\t}\\
&\q+\int_s^T e^{\int_s^r [2\Th^\Pi_1(\t)+2 \Th^\Pi_2(\t)+\si^2]d\t}\a(r-t_{k})\big[-\Th_1^\Pi(r)^2+R\Th_2^\Pi(r)^2\big]dr,\qq s\in[t_{k+1},t_N];\\
& P(t_{k};s)={1\over {1\over e^{\si^2( t_{k+1}-s)} P(t_{k};t_{k+1})}+\int_s^{t_{k+1}} {1-R\over e^{\si^2 (r-s)} \a(r-t_{k})R}dr},\qq s\in[t_{k},t_{k+1}].
\end{aligned}
\ee
Denote
\bel{}
\begin{aligned}
&P^\Pi(t,s)=P(t_{k};s),\qq  (t,s)\in\D[t_{k},T]\setminus\D[t_{k+1},T],\\
&\Th^\Pi_1(s)= {P(t_{k};s)\over \a(s-t_{k})},\qq \Th^\Pi_2(s)=-{ P(t_{k};s)\over R\a(s-t_{k})},\qq  s\in[t_{k},t_{k+1}].
\end{aligned}
\ee

\ms

For any $N>0$, $P^\Pi(\cd,\cd)$ and $\Th^\Pi(\cd)$ can be explicitly obtained on $[0,T]$ by induction.
Then the unique solution $P(\cd,\cd)$ of \rf{RE-TIC} can be obtained for $s\in[0,T]$ by
\bel{}
P(t,s)=\lim_{\|\Pi\|\to 0}P^\Pi(t,s),\qq.
\ee
and the equilibrium strategy $\bar\Th_i(\cd,\cd)$ can be obtained by
\bel{eq: approximation of bar Theta}
\bar\Th_1(s)=\lim_{\|\Pi\|\to 0}\Th_1^\Pi(s)=\lim_{\|\Pi\|\to 0}P^\Pi(s,s),
\qq \bar\Th_2(s)=\lim_{\|\Pi\|\to 0}\Th_2^\Pi(s)=-\lim_{\|\Pi\|\to 0}{P^\Pi(s,s)\over R}.
\ee

\ms
Figure \ref{Fig: 6} displays the approximate strategies $\bar\Th_2(s)$ obtained using the above algorithm to visualize the equilibrium closed-loop saddle strategy. The qualitative characteristics of the lobbying strategies are broadly consistent with those of a two-player game with constant discounting, as shown in Figure \ref{Fig: 6}. We compare the equilibrium closed-loop saddle lobbying strategy $\bar\Th_2(s)$ to the equilibrium lobbying strategy $\ti \Th_2(s)$ for a single player with non constant discounting by plotting them together in Figure \ref{Fig: 7}. The figure clearly demonstrates that lobbying intensifies in a two-player game setting, indicating that strategic interaction between the players leads to increased lobbying efforts, even when the discount rate is non-constant. This general observation holds true across all panels in Figure \ref{Fig: 7}, producing results similar to those obtained when the discount rate is constant and summarized in Figure \ref{Fig: 5}. Figure \ref{Fig: 8} further confirms the absence of "overshooting" in the equilibrium, closed-loop Nash lobbying strategy with non-constant discount remains bounded by the closed-loop Nash lobbying strategies produced under the assumption of constant discounting. This holds true for both scenarios, when the short-term self is in control and when the long-term self is in control.

\begin{figure}[h!]
  \includegraphics[width=1.0\textwidth]{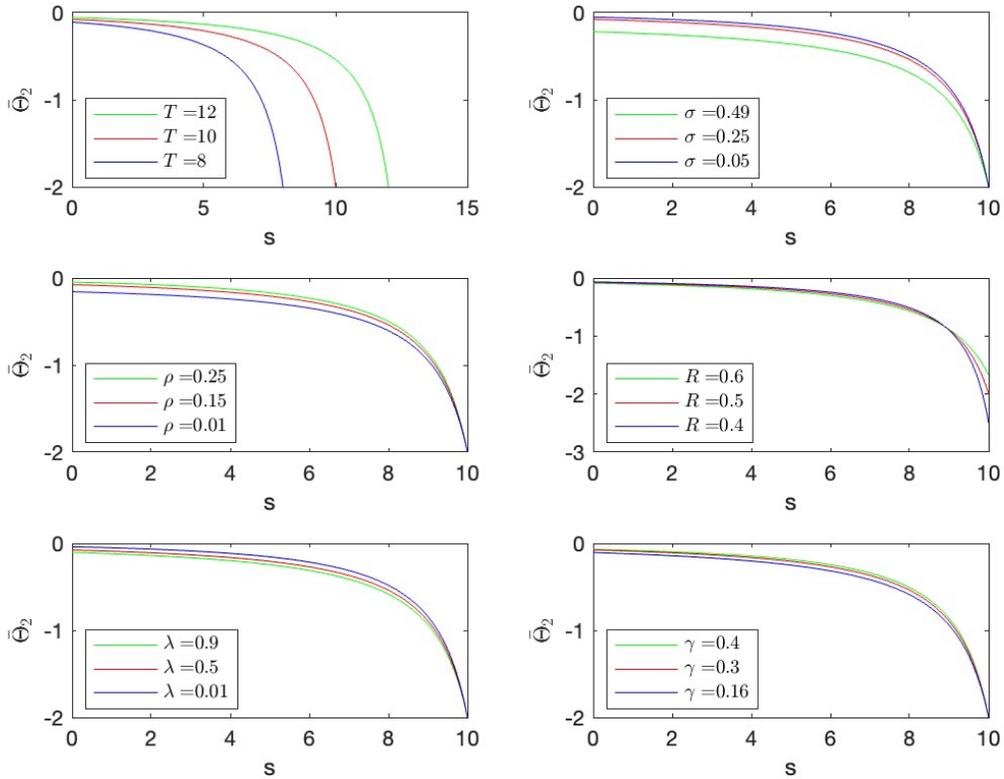}
   \caption{The figure panels display the equilibrium closed-loop saddle strategies $\bar{\Theta}_2(s)$ for Problem (G) with non-constant discounting, plotted for $0 \les s \les T$ using the approximation given by equation \reff{eq: approximation of bar Theta}. We use the baseline parameter values of $T = 10$, $\sigma = 0.25$, $\rho = 0.15$, $R = 0.5$, $\lambda = 0.5$, and $\gamma = 0.3$.
  } \label{Fig: 6}
\end{figure}

\begin{figure}[h!]
  \includegraphics[width=1.0\textwidth]{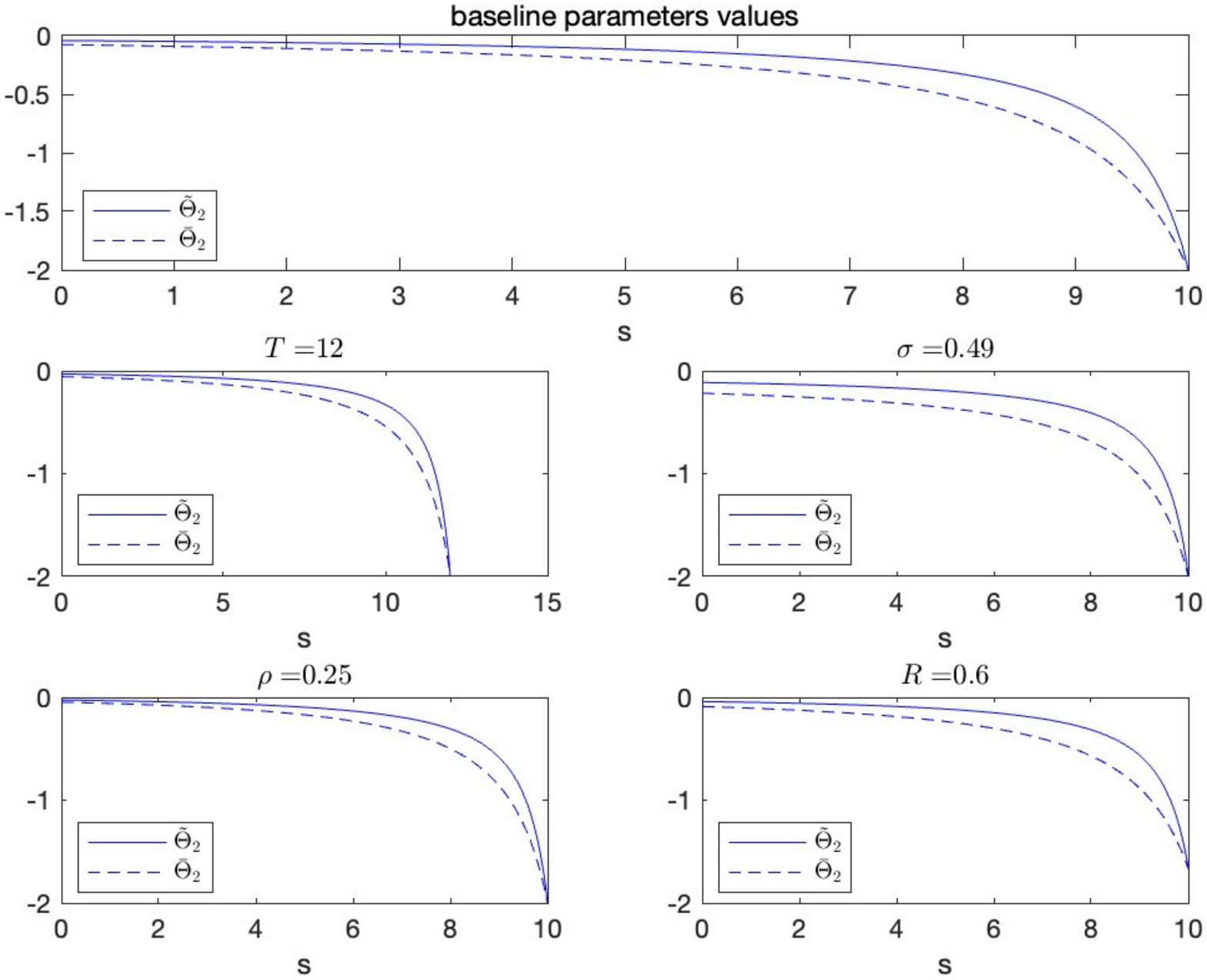}
   \caption{The figure panels display the equilibrium closed-loop saddle strategies $\bar{\Theta}_2(s)$ for Problem (G) with non-constant discounting, as well as the equilibrium closed-loop strategies $\tilde{\Theta}_2(s)$ for the single player with non-constant discounting, plotted for $0 \les s \les T$. The figures are based on the approximations given by equations \reff{eq: approximation of bar Theta} and \reff{eq: approximation single player}. We use the baseline parameter values of $T = 10$, $\sigma = 0.25$, $\rho = 0.15$, $R = 0.5$, $\lambda = 0.5$, and $\gamma = 0.3$.
  } \label{Fig: 7}
\end{figure}

\begin{figure}[h!]
  \includegraphics[width=.85\textwidth]{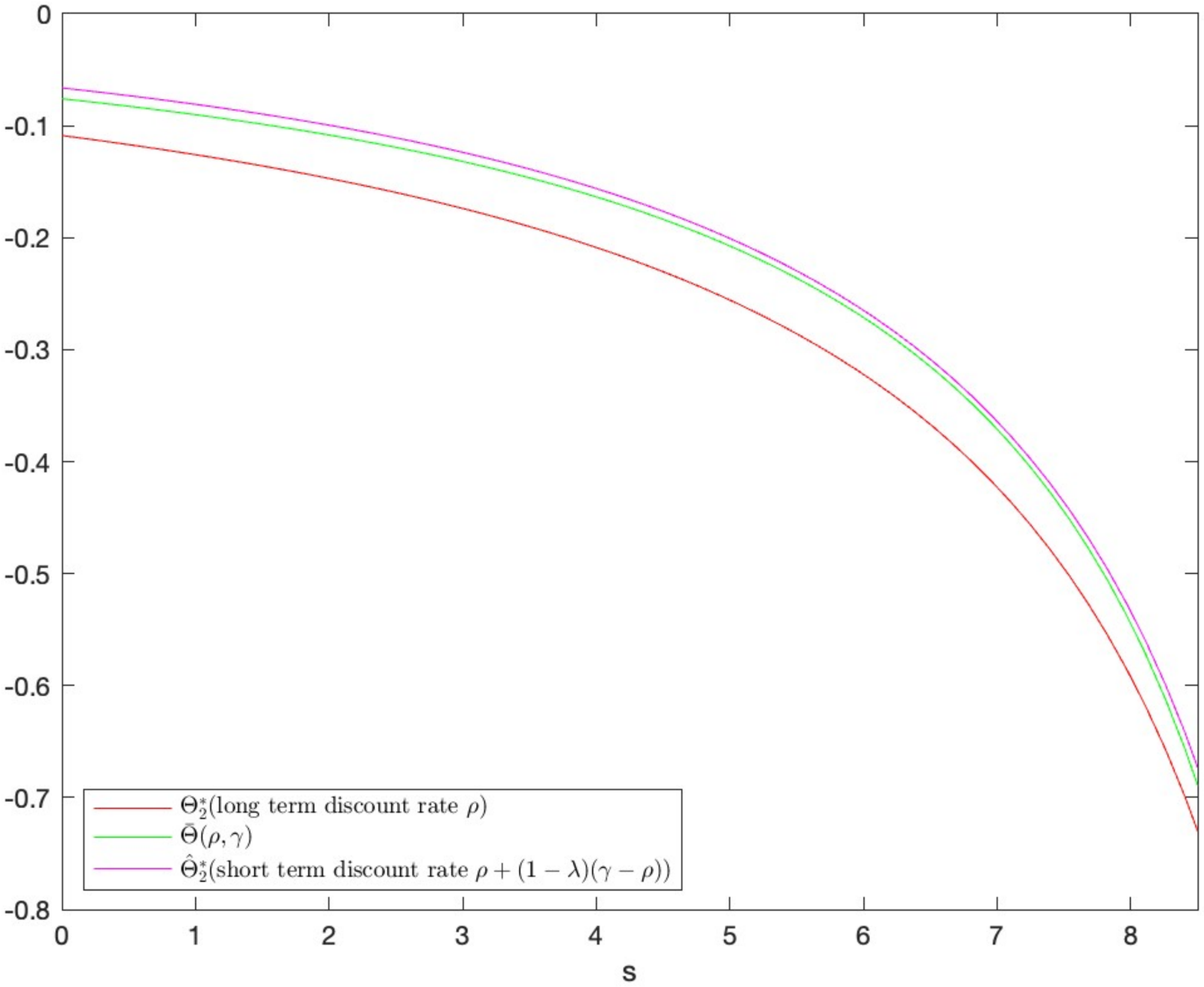}
   \caption{The figure panels display the equilibrium saddle strategies $\bar{\Theta}_2(s)$ for Problem (G) with non-constant discounting, and the closed-loop saddle strategies $\Theta_2^*(s)$ for Problem (G) with constant discounting, plotted for $0 \les s \les T$. These strategies  $\bar{\Theta}_2(s)$ are approximated using the expression \reff{eq: approximation of bar Theta} and the strategies $\Theta_2^*(s)$ are given in the closed-form expression \reff{eq: Nash for Problem N}. We use the baseline parameter values $T = 10$, $\sigma = 0.25$, $\rho = 0.15$, $R = 0.5$, $\lambda = 0.5$, and $\gamma = 0.3$. To visualize the wedge between curves, we truncated the x-axis at $s=8.5$. We have verified numerically that there is no overshooting in the interval $s \in [8.5,10]$.
  } \label{Fig: 8}
\end{figure}

\newpage

\subsection{Conclusion}
We investigate a linear quadratic stochastic zero-sum game in which two players lobby a political representative to invest in a wind turbine farm. As the players have time-inconsistent preferences, they discount short-term performance with a large discount rate and long-term performance with a low discount rate. Our aim is to determine the equilibrium lobbying behavior of the players in both single-player and two-player frameworks.

We find that the equilibrium behavior takes the form of a closed-loop control, which is provided in closed-form or in a form that can be approximated. Our analysis reveals that non-constant discounting does not significantly alter the static comparative with respect to most model parameters, except for the cost of lobbying parameter. We show that strategic behavior can change the monotonicity of the equilibrium lobbying with respect to that parameter. In the single-player framework, a larger cost of lobbying always implies less intense lobbying, while in a two players game it can imply more intense lobbying in some parameter regions.

Our study has revealed that, despite the aforementioned findings, strategic behavior consistently leads to increased lobbying intensity in all situations. Additionally, we have shown that there is no overshooting phenomenon. Specifically, we have found that the equilibrium behavior under non-constant discounting is constrained from below by lobbying behavior when players prioritize long-term performance at a high discount rate, and from above by lobbying behavior when players prioritize short-term performance at a low discount rate. These results suggest that strategic behavior consistently leads to heightened lobbying intensity, while the degree of discounting impacts the upper and lower bounds of lobbying behavior.

We conducted an initial study on the interplay between strategic behavior of two players and Stackelberg behavior of multiple selves induced by non-constant discounting in a zero-sum game. To facilitate a broader understanding of the important topic of dynamic games and time-inconsistency, we intentionally simplified the model and focused on specific cases for our analysis. Additionally, we discussed the questions in the context of lobbying and employed a basic reduced form political economy model to maintain tractability and the possibility of closed-form solutions. However, there are numerous opportunities for expansion, including more extensive modeling of political economy elements. To the best of our knowledge, this paper is the first to apply stochastic differential dynamic games methods in political analysis and future extensions could provide further insights into the dynamics of lobbying and other key interest concepts  that shape the landscape of political economy..

It's worth noting that our framework's zero-sum property reduces the equilibrium description to one Ricatti equation. However, in the interesting case where the players have varying discount functions, the zero-sum assumption falls apart. Similarly, from the vantage point of the lobbying model, this assumption also curtails the full spectrum of externalities one lobbyist might exert on another. Thus, future studies can explore ways to overcome this limitation and expand the applicability of our approach to a broader range of non-zero sum differential games. If possible, such a result would be a significant advance because it will allow to analyze strategic interactions  between players with various degree of behavioral biases.

\section{Appendix}

\subsection{Proof of Proposition \ref{eq: single player optimal control}}

By \cite{yong1999stochastic}, the  Riccati equation associated with the above problem reads
\bel{}\label{eq: Riccati1}\left\{\begin{aligned}
  & \dot{P}(s)+\si^2 P(s)-\rho P(s)-{P(s)^2\over R}=0, \\
  & P(T)=1.
\end{aligned}\right.\ee
characterizes the optimal control associated to the objective \rf{LQC2:cost} and the state \rf{LQC2-SDE}. Solving the equation \reff{eq: Riccati1} gives
\bel{}
P(s)={R(\rho-\si^2) e^{(\rho-\si^2) s}\over [1+R(\rho-\si^2) ]e^{(\rho-\si^2)  T}-e^{(\rho-\si^2)  s}},\qq s\in[0,T].
\ee
Then the unique optimal closed-loop strategy is given by
\bel{}
\hat \Th_2(s)=-{P(s)\over R}=-{(\rho-\si^2)  e^{(\rho-\si^2)  s}\over (1+R(\rho-\si^2) )e^{(\rho-\si^2)  T}-e^{(\rho-\si^2)  s}},\qq s\in[0,T].
\ee
The rest can be obtained easily. \qed

\subsection{Proof of Theorem 1}
Denote the discount function of the precommitment model by
$$
\a^\Pi(t,s)=\left\{ \begin{aligned}
&\a(s-t_{N-1}),\qq t\in[t_{N-1},t_N],\q s\in[t,T],\\
&\a(s-t_{k-1}),\qq t\in[t_{k-1},t_k),\q s\in[t,T],\q k=1,...,N-1.
\end{aligned}\right.
$$
Note that the function $P^\Pi(\cd,\cd)$ defined by \rf{Algorithm step 4} satisfies the following equation:

\bel{ERIC11-Pi}\left\{\begin{aligned}
  & P^\Pi_s(t,s)+2 P^\Pi(t,s)\Th^\Pi_2(s) +P^\Pi(t,s)\si^2 \\
  &\q+\a^\Pi(t,s)R\Th_2^\Pi(s)^2=0, \qq (t,s)\in\D^*[0,T],\\
  & P^\Pi(t,T)=\a^\Pi(T-t),\qq t\in[0,T].
\end{aligned}\right.\ee
From the fact that $\a(s-t)\les\a(s-t^\prime)$ for any $0\les t\les t^\prime\les T$ and $s\in[t^\prime, T]$, we have
\bel{Con1-11}
 P^\Pi(t,s)\les P^\Pi(s,s),\qq 0\les t\les s\les T.
\ee
Let $\Xi(\cd)$ be the unique solution to the following Lyapunov equation:
\bel{Wp11}
\dot{\Xi}(s)+\Xi(s)\si^2=0,\qq \Xi(T)=1.
\ee
From \rf{Algorithm step 1}, it is easily checked that
$$
 P^\Pi(t,s)=P(t_{N-1};s)\les \Xi(s)=e^{\si^2 (T-s)},\qq t_{N-1}\les t\les s\les T.
$$
Then, by \rf{Con1-11}, we have
\bel{Con2-11}
 P^\Pi(t,s)\les \Xi(s),\qq 0\les t\les T,\q t\vee t_{N-1}\les s\les T.
\ee
In particular,
$$P^\Pi(t,t_{N-1})\les\Xi(t_{N-1}),\qq 0\les t\les t_{N-1}.$$
From \rf{Algorithm step 3}, we have
$$
 \begin{aligned}
 P^\Pi(t,s)=P(t_{N-2};s)\les  e^{\si^2 (t_{N-1}-s)}P(t_{N-2};t_{N-1})\les  e^{\si^2 (t_{N-1}-s)}\Xi(t_{N-1})= \Xi(s),
 \\ t_{N-2}\les t< t_{N-1}, \q t\les s\les t_{N-1}.
\end{aligned}
$$
It then follows from \rf{Con1} that
\bel{Con3-11}
 P^\Pi(t,s)\les \Xi(s),\qq 0\les t< t_{N-1},\q t\vee t_{N-2}\les s\les t_{N-1}.
\ee
Thus, by induction, we have
\bel{}
 P^\Pi(t,s)\les \Xi(s),\qq 0\les t\les T,\q t\vee t_{N-2}\les s\les T.
\ee
Continuing the above together yields that
\bel{}
 P^\Pi(t,s)\les \Xi(s),\qq 0\les t\les s\les T.
\ee
Thus, noting $P^\Pi(t,s)\ges 0$, $P^\Pi(\cd,\cd)$ is uniformly bounded.
Then from  \rf{RE-TIC} and \rf{ERIC11-Pi}, we have
$$
|P(t,s)-P^\Pi(t,s)|\les K\|\Pi\|+\int_s^T\[|P(t,r)-P^\Pi(t,r)|+|P(r,r)-P^\Pi(r,r)|\]dr,
$$
which implies that
$$
|P(t,s)-P^\Pi(t,s)|\les K\|\Pi\|.
$$

\subsection{Proof of Proposition \ref{thm:LQG-closed-Nash-kehua}}

Using the fact $G=1$, $R_1(\cd)=-1$ and $R_2(\cd)=R$, we can rewrite \rf{Ric} as follows:
\bel{Ric-1}\left\{\begin{aligned}
  & \dot{P}(s)+\si^2 P(s)-\rho P(s) -{1-R\over R}P(s)^2=0, \\
  & P(T)=1,
\end{aligned}\right.\ee
Note that ${1-R\over R}>0$. Then by \cite{yong1999stochastic},
the ODE \rf{Ric-1} admits a unique solution $P\in C([0,T])$.
Indeed, \rf{Ric-1} can be solved explicitly, with the solution given by
\bel{}
P(s)={R(\rho-\si^2) e^{(\rho-\si^2) s}\over [1-R+R(\rho-\si^2) ]e^{(\rho-\si^2)  T}-(1-R)e^{(\rho-\si^2)  s}},\qq s\in[0,T].
\ee
Substituting the above into \rf{Th^*}, we get the desired results immediately. \qed

\subsection{Proof of Theorem \ref{thn:well-ERE}}
By the variation of constants formula, we can see that  ERE \rf{ERIC} is  essentially a Volterra integral equation.
Notice that the map $t\mapsto P(t,s)$ (if exists) is differentiable.
Then, to prove Theorem \rf{thn:well-ERE}, we introduce  the following Volterra differential-integral equation  (VDIE, for short):
\bel{VDI}\left\{
\begin{aligned}
&\dot{\G}(t)+\G(t)\si^2-{1-R\over R}\G(t)^2 -\int_t^T e^{\int_t^s [\si^2+2{R-1\over R}] \G(r)dr}\partial_t\a(s-t){1-R\over R}\G(s)^2ds\\
&\q-e^{\int_t^T [\si^2+2{R-1\over R}] \G(r)dr}\partial_t\a(T-t)=0,\\
&\G(T)=1.
\end{aligned}\right.
\ee
If \rf{VDI} has a solution, it is easily checked that
$$\begin{aligned}
&{1-R\over R}\G(t)^2 +\int_t^T e^{\int_t^s [\si^2+2{R-1\over R}] \G(r)dr}\partial_t\a(s-t){1-R\over R}\G(s)^2ds\\
&\q+e^{\int_t^T [\si^2+2{R-1\over R}] \G(r)dr}\partial_t\a(T-t)\ges 0,
\end{aligned}$$
Thus, by the comparison theorem of ordinary differential equations (ODEs, for short),
we have the following prior estimate:
\bel{PE-1}
\G(s)\les\Xi(s),\qq s\in[0,T],
\ee
where $\Xi(\cd)$ is the unique solution to the following Lyapunov equation:
\bel{Wp1}
\dot{\Xi}(s)+\Xi(s)\si^2=0,\qq \Xi(T)=1.
\ee

Next, we give a  prior lower bound estimate for $\G(\cd)$.
First, \rf{VDI} can be rewritten as
$$\left\{
\begin{aligned}
&\dot{\G}(t)+[\si^2-2{1-R\over R}\G(t)]\G(t) +{1-R\over R}\G(t)^2\\
&\q -\int_t^T e^{\int_t^s [\si^2-2{1-R\over R}\G(r)]dr}\partial_t\a(s-t){1-R\over R}\G(s)^2ds\\
&\q-e^{\int_t^T [\si^2-2{1-R\over R}\G(r)]dr}\partial_t\a(T-t)=0,\\
&\G(T)=1.
\end{aligned}\right.
$$
Denote
$$
\begin{aligned}
\h\G(t)&\deq\int_t^T e^{\int_t^s [\si^2-2{1-R\over R}\G(r)]dr}\a(s-t){1-R\over R}\G(s)^2ds\\
&\q+e^{\int_t^T [\si^2-2{1-R\over R}\G(r)]dr}\a(T-t),\qq t\in[0,T].
\end{aligned}
$$
Clearly,
$$
\h\G(t)\ges0,\qq t\in[0,T].
$$
Note that
$$\left\{
\begin{aligned}
&\dot{\h\G}(t)+[\si^2-2{1-R\over R}\G(t)]\h\G(t) +{1-R\over R}\G(t)^2\\
&\q -\int_t^T e^{\int_t^s [\si^2-2{1-R\over R}\G(r)]dr}\partial_t\a(s-t){1-R\over R}\G(s)^2ds\\
&\q-e^{\int_t^T [\si^2-2{1-R\over R}\G(r)]dr}\partial_t\a(T-t)=0,\\
&\h\G(T)=1.
\end{aligned}\right.
$$
Then by the uniqueness of the solution to the above linear equation with unknown $\h\G(\cd)$, we have
$$
\G(t)=\h\G(t),\qq t\in[0,T].
$$
It follows that
$$
\G(t)\ges0,\qq t\in[0,T].
$$
Combining this with \rf{PE-1}, we get the following prior estimate:
\bel{PE}
0\les\G(s)\les\Xi(s),\qq s\in[0,T].
\ee
With the above, we can prove the well-posedness of  \rf{VDI}.
\begin{lem}\label{Well-DIV}
 The VDIE \rf{VDI} admits a unique solution $\G(\cd)\in C([0,T])$.
Moreover, $\G(t)\ges 0,\, t\in[0,T]$.
\end{lem}

\begin{proof} {\bf of Lemma \ref{Well-DIV}}
The uniqueness of solutions to VDIE \rf{VDI} can be obtained by a standard method,
and we only prove the existence of a solution.
Let $M\deq \sup_{t\in[0,T]}|\Xi(t)|$. Let $\rho_M(\cd)$ be a smooth truncation function with
$$\rho_M(x)=\left\{
\begin{aligned}
&x,\qq |x|\les M+1;\\
&0,\qq |x|\ges M+2.
\end{aligned}\right.
$$
Let $\G^M(\cd)$ be the unique solution to the following differential-integral equation, which satisfies the uniformly Lipschitz condition:
\bel{VDI-M}\left\{
\begin{aligned}
&\dot{\G}^M(t)+\G^M(t)\si^2-{1-R\over R}\rho_M(\G^M(t))^2 \\
&\q-\int_t^T e^{\int_t^s [\si^2+2{R-1\over R}] \rho_M(\G^M(r))dr}\partial_t\a(s-t){1-R\over R}\rho_M(\G^M(s))^2ds\\
&\q-e^{\int_t^T [\si^2+2{R-1\over R}]\rho_M( \G(r))dr}\partial_t\a(T-t)=0,\\
&\G^M(T)=1.
\end{aligned}\right.
\ee
 Denote $\t=\max\{t\in[0,T] \,|\, |\G^M(t)|\ges M+1 \}$, and $\t=0$,
 if $\{t\in[0,T] \,|\, |\G^M(t)|\ges M+1 \}=\emptyset$. Note that $|\G^M(T)|=1\les  \sup_{t\in[0,T]}|\Xi(t)|<M+1$.
 Thus, $\t<T$. If $\t=0$, $\rho_M(\G^M(\cd))=\G^M(\cd)$, which implies that $\G^M(\cd)$ is a solution to \rf{VDI}.
 If $\t>0$, then $|\G^M(\t)|=M+1$, and $\G^M(\cd)$ is a solution to \rf{VDI} on $[\t,T]$.
 However, by the prior estimate \rf{PE}, we have $|\G^M(t)|\les M,\, t\in[\t,T]$,
 which contradicts with $|\G^M(\t)|=M+1$.
 Thus, $\t=0$, and the proof is complete. \qed
\end{proof}

\ms

\noindent
\emph{\textbf{Complete the proof of  Theorem \ref{thn:well-ERE}}}.
We still only prove the existence of a solution here.
Denote
\bel{well-ERE-proof1}
\begin{aligned}
P(t,s)&\deq\int_s^T e^{\int_s^\t [\si^2-2{1-R\over R}\G(r)]dr}\a(\t-t){1-R\over R}\G(\t)^2d\t\\
&\q+e^{\int_s^T [\si^2-2{1-R\over R}\G(r)]dr}\a(T-t),\qq 0\les t\les s\les T.
\end{aligned}
\ee
Then
\bel{well-ERE-proof2}
P(t,t)=\G(t),\qq t\in[0,T].
\ee
and
\bel{well-ERE-proof3}
\left\{\begin{aligned}
  & P_s(t,s)+P(t,s)\big[\si^2-2{1-R\over R}\G(s)\big]+\a(s-t){1-R\over R}\G(s)^2=0, \\
  &\qq\qq\qq\qq\qq\qq\qq\qq\qq\qq 0\les t\les s\les T,\\
  & P(t,T)=\a(T-t), \qq t\in[0,T].
\end{aligned}\right.\ee
Combining \rf{well-ERE-proof2} and \rf{well-ERE-proof3} together,
we can see that the function $P(\cd,\cd)$, defined by  \rf{well-ERE-proof1},
is a solution to \rf{ERIC}.
If $R=1$, then \rf{ERIC} reads
\bel{}\left\{\begin{aligned}
& P_s(t,s)+P(t,s)\si^2 =0, \qq (t,s)\in\D[0,T],\\
&P(t,T)=\a(T-t).
\end{aligned}\right.\ee
Taking $t$ as a parameter, by  the variation
of constants formula, we get  \rf{ERIC} immediately.
This completes the proof.
$\hfill\qed$

\subsection{Proof of Theorem \ref{thm:VT}}

We only prove the first inequality in \rf{def-equilibrium1} holds, and the second one can be obtained similarly.
For any fixed $t\in[0,T)$, consider the following classical game problem over $[t,t+\e]$ with the state
\bel{SDE-e-game}\left\{\2n\ba{ll}
\ns\ds dX^\e(s)=\big[ u_1(s)+ u_2(s)\big]ds+\si X^\e(s)dW(s),\qq s\in[t,t+\e],\\
\ns\ds X^\e(t)=\bar X(t).\ea\right.\ee
and the functionals
\bel{}\begin{aligned}
&J_1(t,\bar X(t);u_1(\cd),u_2(\cd))|_{[t,t+\e]}\deq J_1(t,\bar X(t);\Th_1^\e(\cd)X^\e(\cd),\Th^\e_2(\cd)X^\e(\cd))\\
&\q=\dbE_t\Big\{P(t,t+\e)X^\e(t+\e)^2+\int_t^{t+\e}\a(s-t)\(  - u_1(s)^2+R u_2(s)^2\)ds\Big\},\\
&J_2(t,\bar X(t);u_1(\cd),u_2(\cd))|_{[t,t+\e]}\deq J_2(t,\bar X(t);\Th_1^\e(\cd)X^\e(\cd),\Th^\e_2(\cd)X^\e(\cd))\\
&\q=\dbE_t\Big\{-P(t,t+\e)X^\e(t+\e)^2+\int_t^{t+\e}\a(s-t)\(  u_1(s)^2-R u_2(s)^2\)ds\Big\},
\end{aligned}\ee
where $\Th_i^\e(\cd)$ is defined by \rf{def-equilibrium2}.
Then by Proposition \ref{thm:LQG-closed-Nash-kehua}, the unique closed-loop Nash strategy is given by
\bel{}
\Th_1^{*,\e}(s)=P^t(s),\qq \Th_w^{*,\e}(s)=-{P^t(s)\over R},\qq s\in[t,t+\e],
\ee
where $P^t(\cd)$ is the unique solution to the following Riccati equation:
\bel{Ric-e}\left\{\begin{aligned}
  & \dot{P}^t(s)+P^t(s)\si^2 - {1-R\over \a(s-t)R}P^t(s)^2=0,\qq s\in[t,t+\e],\\
  & P^t(t+\e)= P(t,t+\e).
\end{aligned}\right.\ee
In other words, for any $(u_1(\cd),u_2(\cd))\in L^2(t,t+\e) \times L^2(t,t+\e)$,
 the following holds:
\bel{Veri-Th-proof5}\ba{ll}
\ns\ds J(t,\bar X(t);\Th^{*,\e}_1(\cd) X^{*,\e},\Th^{*,\e}_2(\cd) X^{*,\e})|_{[t,t+\e]}\ges J(t,\bar X(t); u_1(\cd),\Th^{*,\e}_2(\cd) X^{1,\e} (\cd))|_{[t,t+\e]}, \\
\ns\ds  J(t,\bar X(t);\Th^{*,\e}_1(\cd) X^{*,\e},\Th^{*,\e}_2(\cd) X^{*,\e}) |_{[t,t+\e]}\les J(t,\bar X(t);\Th^{*,\e}_1(\cd) X^{2,\e} (\cd),u_2(\cd))|_{[t,t+\e]},\ea\ee
where
\bel{}\left\{\2n\ba{ll}
\ns\ds dX^{*,\e}(s)=\big[P^t(s)X^{*,\e}(s))-R^{-1} P^t(s)X^{*,\e}(s)\big]ds+\si X^{*,\e}(s) dW(s),\\
\ns\ds X^{*,\e}(t)=\bar X(t),\ea\right.\ee

\bel{Veri-Th-proof7}\left\{\2n\ba{ll}
\ns\ds dX^{1,\e}(s)=\big[ u_1(s)-R^{-1} P^t(s)X^{1,\e}(s)\big]ds+\si X^{1,\e}(s) dW(s),\\
\ns\ds X^{1,\e}(t)=\bar X(t),\ea\right. \ee
and
\bel{Veri-Th-proof7-1}\left\{\2n\ba{ll}
\ns\ds dX^{2,\e}(s)=\big[ P^t(s)X^{2,\e}(s)+u^2(s)\big]ds+\si X^{2,\e}(s) dW(s),\\
\ns\ds X^{2,\e}(t)=\bar X(t).\ea\right. \ee
Note that
\bel{Veri-Th-proof1}
J(t,\bar X(t);\Th^{*,\e}_1(\cd) X^{*,\e},\Th^{*,\e}_2(\cd) X^{*,\e})|_{[t,t+\e]}=P^t(t)\bar X(t)^2,
\ee
and
\bel{Veri-Th-proof2}
J(t,\bar X(t);\bar\Th_1(\cd)\bar X(\cd),\bar\Th_2(\cd)\bar X(\cd))=P(t,t)\bar X(t)^2.
\ee
Clearly,
\bel{Veri-Th-proof8}
\begin{aligned}
|P^t(s)-P(t,s)|&\les |P^t(s)-P(t,t+\e)|+|P(t,t+\e)-P(t,s)|\les K\e,\\
|P^t(s)-P(s,s)|&\les |P^t(s)-P(t,s)|+|P(t,s)-P(s,s)|\les K\e.
\end{aligned}
\ee
Then from \rf{Ric-e} and \rf{ERIC}, we get
\bel{Veri-Th-proof3}
|P^t(t)-P(t,t)|\les K\int_t^{t+\e} \big[|P^t(s)-P(t,s)|+|P^t(s)-P(s,s)|\big]ds\les K\e^2.
\ee
Thus, comparing \rf{Veri-Th-proof1} and \rf{Veri-Th-proof2}, we have
\bel{Veri-Th-proof4}
\big|J(t,\bar X(t);\Th^{*,\e}_1(\cd) X^{*,\e},\Th^{*,\e}_2(\cd) X^{*,\e})|_{[t,t+\e]}-J(t,\bar X(t);\bar\Th_1(\cd)\bar X(\cd),\bar\Th_2(\cd)\bar X(\cd))\big|
\les K\e^2\bar X(t)^2.
\ee
It follows from \rf{Veri-Th-proof5} that
\bel{Veri-Th-proof6}\ba{ll}
\ns\ds J(t,\bar X(t); u_1(\cd),\Th^{*,\e}_2(\cd) X^{1,\e} (\cd))|_{[t,t+\e]}
\les J(t,\bar X(t);\bar\Th_1(\cd)\bar X(\cd),\bar\Th_2(\cd)\bar X(\cd))+K\e^2\bar X(t)^2.\ea\ee
Thus, to prove that \rf{def-equilibrium1} holds with $i=1$, it suffices to prove that
\bel{Veri-Th-proof12}
\begin{aligned}
&\big|J(t,\bar X(t); u_1(\cd),\Th^{*,\e}_2(\cd) X^{1,\e} (\cd))|_{[t,t+\e]}-J(t,\bar X(t);u_1(\cd),\bar\Th_2(\cd)X^\e(\cd))\big|\\
&\q\les K\e^2 \dbE_t\[\bar X(t)^2+\int_t^{t+\e} |u_1(s)|^2ds\]^{1\over 2},
\end{aligned}
\ee
for some $K>0$. Note that
\bel{Veri-Th-proof9}
\begin{aligned}
&J(t,\bar X(t); u_1(\cd),\Th^{*,\e}_2(\cd) X^{1,\e} (\cd))|_{[t,t+\e]}=\dbE_t\Big\{P(t,t+\e)X^{1,\e}(t+\e)^2\\
&\q+\int_t^{t+\e}\a(s-t)\(- u_1(s)^2+R^{-1}P^t(s)^2 X^{1,\e}(s)^2\)ds\Big\},
\end{aligned}
\ee
where $ X^{1,\e} (\cd))$ is the unique solution to \rf{Veri-Th-proof7}, and
\bel{Veri-Th-proof10}
\begin{aligned}
&J(t,\bar X(t);\Th_1^\e(\cd)X^\e(\cd),\bar\Th_2(\cd)X^\e(\cd))=\dbE_t\Big\{P(t,t+\e)X^{\e}(t+\e)^2\\
&\q+\int_t^{t+\e}\a(s-t)\( - u_1(s)^2+R^{-1}P(s,s)^2 X^{\e}(s)^2\)ds\Big\},
\end{aligned}
\ee
where $X^\e(\cd)$ is uniquely determined by
\bel{Veri-Th-proof8}\left\{\2n\ba{ll}
\ns\ds dX^{\e}(s)=\big[ u_1(s)-R^{-1} P(s,s)X^{\e}(s)\big]ds+\si X^{\e}(s) dW(s),\qq s\in[t,t+\e],\\
\ns\ds X^\e(t)=\bar X(t).\ea\right. \ee
Comparing \rf{Veri-Th-proof1} with \rf{Veri-Th-proof8}, by \rf{Veri-Th-proof8} and the stability of SDEs,
we have
\bel{}\begin{aligned}
&\dbE_t\[\sup_{s\in[t,t+\e]}|X^\e(s)-X^{1,\e}(s)|^2\]\les K\dbE_t\(\int_t^{t+\e} |P^t(s)-P(s,s)||X^\e(s)|ds\)^2\\
&\q\les K\e^4 \dbE_t\[\bar X(t)^2+\int_t^{t+\e} |u_1(s)|^2ds\].
\end{aligned}\ee
Then
\bel{Veri-Th-proof11}\begin{aligned}
&\dbE_t\[\sup_{s\in[t,t+\e]}|X^\e(s)^2-X^{1,\e}(s)^2|\]\\
&\q\les K\dbE_t\[\sup_{s\in[t,t+\e]}|X^\e(s)+X^{1,\e}(s)|^2\]^{1\over 2}\dbE_t\[\sup_{s\in[t,t+\e]}|X^\e(s)-X^{1,\e}(s)|^2\]^{1\over 2}\\
&\q\les K\e^2 \dbE_t\[\bar X(t)^2+\int_t^{t+\e} |u_1(s)|^2ds\].
\end{aligned}\ee
Comparing \rf{Veri-Th-proof9} and \rf{Veri-Th-proof10}, by the estimates \rf{Veri-Th-proof8} and \rf{Veri-Th-proof11},
we get \rf{Veri-Th-proof12}.
Combining this with \rf{Veri-Th-proof6}, we have
\bel{}
\begin{aligned}
&J(t,\bar X(t);\Th_1^\e(\cd)X^\e(\cd),\bar\Th_2(\cd)X^\e(\cd))\\
&\q\les  J(t,\bar X(t);\bar\Th_1(\cd)\bar X(\cd),\bar\Th_2(\cd)\bar X(\cd))+K\e^2\dbE_t\[\bar X(t)^2+\int_t^{t+\e} |u_1(s)|^2ds\].
\end{aligned}
\ee
It follows that the first inequality in \rf{def-equilibrium1} holds.
Similarly, we can obtain the second inequality in \rf{def-equilibrium1}.

\subsection{Proof of Theorem \ref{thm:convergence}}

Note that for any $0\les t\les t^\prime\les T$ and $s\in[t^\prime, T]$, $\a(s-t)\les\a(s-t^\prime)$.
Combining this with the fact $-\bar\Th_1(s)^2+R\bar\Th_2(s)^2\ges 0$, from \rf{ERIC-Pi} we have
\bel{Con1}
 P^\Pi(t,s)\les P^\Pi(s,s),\qq 0\les t\les s\les T.
\ee
Recall \rf{Ric-N}. By the comparison theorem of ODEs, we have
$$
 P^\Pi(t,s)=P(t_{N-1};s)\les \Xi(s),\qq t_{N-1}\les t\les s\les T,
$$
where $\Xi(\cd)$ is the unique solution to \rf{Wp1}.
Then, by \rf{Con1}, we have
\bel{Con2}
 P^\Pi(t,s)\les \Xi(s),\qq 0\les t\les T,\q t\vee t_{N-1}\les s\les T.
\ee
In particular,
$$P^\Pi(t,t_{N-1})\les\Xi(t_{N-1}),\qq 0\les t\les t_{N-1}.$$
Recall \rf{Ric-(N-1)}. By the comparison theorem  again, we have
 $$
 P^\Pi(t,s)=P(t_{N-2};s) \les \Xi(s),\qq t_{N-2}\les t< t_{N-1}, \q t\les s\les t_{N-1}.
$$
By \rf{Con1}, we have
\bel{Con3}
 P^\Pi(t,s)\les \Xi(s),\qq 0\les t< t_{N-1},\q t\vee t_{N-2}\les s\les t_{N-1}.
\ee
Combining \rf{Con2} and \rf{Con3} together, we have
\bel{Con4}
 P^\Pi(t,s)\les \Xi(s),\qq 0\les t\les T,\q t\vee t_{N-2}\les s\les T.
\ee
By continuing the above, we have
\bel{Con5}
 P^\Pi(t,s)\les \Xi(s),\qq 0\les t\les s\les T.
\ee
Thus, noting $P^\Pi(t,s)\ges 0$, $P^\Pi(\cd,\cd)$ is uniformly bounded.
Then from \rf{ERIC-Pi} and \rf{ERIC}, we have
$$
|P(t,s)-P^\Pi(t,s)|\les K\|\Pi\|+\int_s^T\[|P(t,r)-P^\Pi(t,r)|+|P(r,r)-P^\Pi(r,r)|\]dr,
$$
which implies that
$$
|P(t,s)-P^\Pi(t,s)|\les K\|\Pi\|.
$$
The desired results can be directly obtained.


\clearpage
\newpage
\addcontentsline{toc}{section}{References}
\bibliographystyle{jf}
\openup -.08in
\bibliography{reference}

\ms


\end{document}